%% file: main_arxiv.tex
\documentclass[11pt]{article}

\input{preamble/common}
\input{preamble/arxiv}

\title{Universal Machine-learning Molecular Dynamics at the Speed of Empirical Potentials}

\author{
    \begin{minipage}{0.95\textwidth}
        \centering
        Tiancheng Li$^{1,2}$ \quad
        Jianming Xue$^{1,3,*}$ \quad
        Linfeng Zhang$^{2,4,*}$\\
        Duo Zhang$^{2,4,5,*}$ \quad
        Han Wang$^{3,6,*}$ \\[0.6em]
        {\normalfont\small
            $^{1}$State Key Laboratory of Nuclear Physics and Technology, School of Physics, Peking University, Beijing 100871, China\\
            $^{2}$AI for Science Institute, Beijing 100080, P. R. China\\
            $^{3}$HEDPS, CAPT, College of Engineering, Peking University, Beijing 100871, P. R. China\\
            $^{4}$DP Technology, Beijing 100080, P. R. China\\
            $^{5}$Academy for Advanced Interdisciplinary Studies, Peking University, Beijing 100871, P. R. China\\
            $^{6}$National Key Laboratory of Computational Physics, Institute of Applied Physics and Computational Mathematics, Fenghao East Road 2, Beijing 100094, P. R. China\\
            $^{*}$Correspondence: \href{mailto:jmxue@pku.edu.cn}{jmxue@pku.edu.cn};
            \href{mailto:linfeng.zhang.zlf@gmail.com}{linfeng.zhang.zlf@gmail.com};
            \href{mailto:zhduodyx@pku.edu.cn}{zhduodyx@pku.edu.cn};
            \href{mailto:wang_han@iapcm.ac.cn}{wang\_han@iapcm.ac.cn}
        }
    \end{minipage}
}

\date{}

\begin{document}

\maketitle

\addtocontents{toc}{\protect\setcounter{tocdepth}{-1}}

\input{chap/abstract}
\input{chap/introduction}
\input{chap/results}
\input{chap/discussion}
\input{chap/methods}
\input{chap/acknowledgments}
\input{chap/data_availability}

\bibliographystyle{naturemag-doi}
\bibliography{ref}

\input{chap/si}

\end{document}

%% file: preamble/common.tex
\usepackage[utf8]{inputenc}
\usepackage[T1]{fontenc}
\usepackage{amsmath}
\usepackage{amsthm}
\usepackage{amsfonts}
\usepackage{amssymb}
\usepackage{mathtools}
\usepackage{xcolor}
\usepackage{graphicx}
\usepackage[version=4]{mhchem}
\usepackage[flushleft]{threeparttable}
\usepackage{booktabs}
\usepackage{makecell}
\usepackage{siunitx}
\usepackage{setspace}
\usepackage{array}
\usepackage{enumitem}
\usepackage{bm}
\usepackage{wrapfig}
\usepackage{placeins}
\usepackage{fix-cm}
\usepackage{tocloft}
\usepackage{listings}

\newcommand{\PaperTableStyle}{%
    \centering
    \footnotesize
    \setlength{\tabcolsep}{3pt}%
    \renewcommand{\arraystretch}{1.10}%
}
\newcommand{\PaperTableNotesStyle}{\scriptsize}

\lstdefinestyle{LAMMPSInput}{%
    basicstyle=\ttfamily\scriptsize,
    columns=fullflexible,
    keepspaces=true,
    showstringspaces=false,
    breaklines=true,
    frame=tb,
    framerule=0.4pt,
    rulecolor=\color{black},
    captionpos=t,
    aboveskip=0.75em,
    belowskip=0.75em
}

\renewcommand{\thefootnote}{\fnsymbol{footnote}}

\newcommand{\PreserveBackslash}[1]{\let\temp=\\#1\let\\=\temp}
\newcolumntype{C}[1]{>{\PreserveBackslash\centering}p{#1}}
\newcolumntype{R}[1]{>{\PreserveBackslash\raggedleft}p{#1}}
\newcolumntype{L}[1]{>{\PreserveBackslash\raggedright}p{#1}}

\providecommand{\Orth}{\operatorname{O}}
\providecommand{\SO}{\operatorname{SO}}
\providecommand{\STF}{\operatorname{STF}}

\providecommand{\sinc}{\operatorname{sinc}}
\providecommand{\SiLU}{\operatorname{SiLU}}
\providecommand{\SwiGLU}{\operatorname{SwiGLU}}

\providecommand{\vech}{\operatorname{vech}}
\providecommand{\tr}{\operatorname{tr}}
\providecommand{\clip}{\operatorname{clip}}

\providecommand{\MLP}{\operatorname{MLP}}
\providecommand{\tp}{^{\mathsf{T}}}

\providecommand{\R}{\mathbb{R}}

\providecommand{\Sph}{\mathbb{S}^{2}}

\providecommand{\bA}{\bm{A}}
\providecommand{\bB}{\bm{B}}
\providecommand{\bD}{\bm{D}}
\providecommand{\bF}{\bm{F}}
\providecommand{\bG}{\bm{G}}
\providecommand{\bI}{\bm{I}}
\providecommand{\bJ}{\bm{J}}
\providecommand{\bQ}{\bm{Q}}
\providecommand{\bR}{\bm{R}}

\providecommand{\bU}{\bm{U}}
\providecommand{\bW}{\bm{W}}
\providecommand{\bX}{\bm{X}}
\providecommand{\bZ}{\bm{Z}}

\providecommand{\bff}{\bm{f}}
\providecommand{\bg}{\bm{g}}
\providecommand{\bh}{\bm{h}}
\providecommand{\bq}{\bm{q}}
\providecommand{\br}{\bm{r}}
\providecommand{\bu}{\bm{u}}
\providecommand{\bv}{\bm{v}}
\providecommand{\bz}{\bm{z}}

\providecommand{\bbeta}{\bm{\beta}}
\providecommand{\bgamma}{\bm{\gamma}}
\providecommand{\bpsi}{\bm{\psi}}
\providecommand{\bPi}{\bm{\Pi}}
\providecommand{\bXi}{\bm{\Xi}}

\DeclareSIUnit{\Angstrom}{\text{\AA}}
\providecommand{\Ang}{\unit{\Angstrom}}

\providecommand{\rcut}{r_{\mathrm{c}}}

\providecommand{\Nb}{\mathcal{N}}

\theoremstyle{plain}
\newtheorem{theorem}{Theorem}[section]
\newtheorem{proposition}[theorem]{Proposition}

\theoremstyle{definition}

\theoremstyle{remark}

\newcommand{\beginsupplement}{%
    \clearpage
    \setcounter{section}{0}%
    \renewcommand{\thesection}{S\arabic{section}}%
    \renewcommand{\theHsection}{S.\arabic{section}}%
    \setcounter{equation}{0}%
    \renewcommand{\theequation}{S\arabic{equation}}%
    \renewcommand{\theHequation}{S.\arabic{equation}}%
    \setcounter{figure}{0}%
    \renewcommand{\thefigure}{S\arabic{figure}}%
    \renewcommand{\theHfigure}{S.\arabic{figure}}%
    \setcounter{table}{0}%
    \renewcommand{\thetable}{S\arabic{table}}%
    \renewcommand{\theHtable}{S.\arabic{table}}%
    \setcounter{lstlisting}{0}%
    \renewcommand{\thelstlisting}{S\arabic{lstlisting}}%
}

%% file: preamble/arxiv.tex
\usepackage[letterpaper]{geometry}
\usepackage[superscript,biblabel,nomove]{cite}
\usepackage[colorlinks=true,allcolors=blue]{hyperref}
\usepackage{mathptmx}
\usepackage{microtype}

\makeatletter
\renewcommand{\section}{%
    \@startsection{section}{1}{\z@}%
    {-2.0ex \@plus -0.5ex \@minus -0.2ex}%
    {1.5ex \@plus 0.3ex \@minus 0.2ex}%
    {\large\bfseries\raggedright}%
}
\renewcommand{\subsection}{%
    \@startsection{subsection}{2}{\z@}%
    {-1.8ex \@plus -0.5ex \@minus -0.2ex}%
    {0.8ex \@plus 0.2ex}%
    {\normalsize\bfseries\raggedright}%
}
\renewcommand{\subsubsection}{%
    \@startsection{subsubsection}{3}{\z@}%
    {-1.5ex \@plus -0.5ex \@minus -0.2ex}%
    {0.5ex \@plus 0.2ex}%
    {\normalsize\bfseries\raggedright}%
}
\renewcommand{\paragraph}{%
    \@startsection{paragraph}{4}{\z@}%
    {1.5ex \@plus 0.5ex \@minus 0.2ex}%
    {-1em}%
    {\normalsize\bfseries}%
}

\newcommand{\@toptitlebar}{%
    \hrule height 4pt
    \vskip 0.25in
    \vskip -\parskip
}
\newcommand{\@bottomtitlebar}{%
    \vskip 0.29in
    \vskip -\parskip
    \hrule height 1pt
    \vskip 0.09in
}

\renewcommand{\maketitle}{%
    \par
    \begingroup
    \renewcommand{\thefootnote}{\fnsymbol{footnote}}%
    \renewcommand{\@makefnmark}{\hbox to \z@{$^{\@thefnmark}$\hss}}%
    \long\def\@makefntext##1{%
        \parindent 1em\noindent
        \hbox to 1.8em{\hss $\m@th ^{\@thefnmark}$}##1%
    }%
    \thispagestyle{empty}%
    \vbox{%
        \hsize\textwidth
        \linewidth\hsize
        \vskip 0.1in
        \@toptitlebar
        \centering
        {\LARGE\bfseries \@title\par}
        \@bottomtitlebar
        \def\And{%
            \end{tabular}\hfil\linebreak[0]\hfil
            \begin{tabular}[t]{c}\bfseries\rule{\z@}{24pt}\ignorespaces
        }%
        \def\AND{%
            \end{tabular}\hfil\linebreak[4]\hfil
            \begin{tabular}[t]{c}\bfseries\rule{\z@}{24pt}\ignorespaces
        }%
        \begin{tabular}[t]{c}\bfseries\rule{\z@}{24pt}\@author\end{tabular}%
        \vskip 0.3in \@minus 0.1in
    }%
    \@thanks
    \endgroup
    \let\maketitle\relax
    \let\thanks\relax
}

\renewenvironment{abstract}%
{%
    \vskip 0.075in
    \centerline{\large\bfseries Abstract}
    \vspace{0.5ex}
    \begin{quote}
}
{
    \par
    \end{quote}
    \vskip 1ex
}
\makeatother

%% file: chap/abstract.tex
\begin{abstract}
    No interatomic potential has offered universality across chemistry, near-first-principles accuracy and the speed of empirical potentials at once.
    Here we introduce DPA4C, an equivariant potential whose architecture and compressed CUDA operators are co-designed under deployment constraints to pursue accuracy and efficiency together.
    Five variants spanning a 49-fold parameter range form the high-throughput end of the measured accuracy--throughput frontier.
    The largest variant approaches the accuracy of the MACE-Omat models at about two orders of magnitude higher measured throughput.
    The most compact reduces the energy, force and stress errors of the fastest existing universal MLIP by 61.4\%, 48.1\% and 34.3\% at 1.92 times its saturated throughput.
    All five variants complete multimillion-atom simulations on a single GPU and run molecular dynamics for 2.048 billion atoms on 1,024 16-GB NVIDIA V100 GPUs at 83.3--91.2\% weak-scaling efficiency.
    Compared with the MEAM empirical potential, DPA4C-Nano reaches 1.8 and 2.5 times the saturated throughput in single-GPU scans on the same V100 hardware for diamond carbon and FCC copper, respectively.
    DPA4C therefore brings quantum-trained universal accuracy into a regime of speed and system size previously associated with empirical potentials.
\end{abstract}

%% file: chap/introduction.tex
\section{Introduction}
\label{sec:introduction}

No available interatomic potential simultaneously delivers universality across chemistry, near-first-principles accuracy and the speed of empirical potentials.
Accuracy and scale have been achieved together by system-specific machine-learning interatomic potentials (MLIPs)~\cite{behler2007generalized,bartok2010gaussian,thompson2015spectral,shapeev2016moment,drautz2019atomic,zhang2018deep}, which exceed one hundred million atoms in optimized implementations~\cite{jia2020pushing} and approach empirical-potential cost after tabulated compression~\cite{lu2022dp}.
Yet each such potential is bound to one composition, whereas the processes that demand large-scale molecular dynamics (MD) are intrinsically multi-component, such as segregation at grain-boundary networks, multi-principal-element alloys and reactive interfaces.
Retraining for every new system is a data-generation campaign of weeks to months~\cite{zhang2019active,podryabinkin2017active,Vandermause2022} that cannot follow the combinatorial growth of composition space.
Universal MLIPs amortize this effort into a single pre-training across broad chemistry~\cite{chen2022universal,deng2023chgnet,batatia2023foundation,yang2024mattersim}, pairing universality with accuracy.
Their computational cost, however, confines them to sizes and rates far below the empirical-potential regime in which those processes occur.

Universality and accuracy are pursued together by equivariant neural networks.
Message-passing models such as NequIP~\cite{batzner2022nequip}, MACE~\cite{batatia2022mace}, Equiformer~\cite{liao2022equiformer,liao2023equiformerv2}, eSEN~\cite{fu2025esen} and DPA4~\cite{li2026dpa4} iterate learned tensor features over the atomistic graph, so the feature width enters communication and persistent memory during MD.
Strictly local variants remove the inter-atomic propagation.
Allegro, for example, maintains learned equivariant tensors on every edge and updates them through layered tensor products~\cite{musaelian2022learning}.
Recent foundation models in this family reach the leading inference speeds among equivariant universal potentials through software optimization of a fixed architecture, combining compiled graphs, accelerated tensor-product kernels and mixed-precision techniques~\cite{kavanagh2026fast}.
Even so, their reported single-GPU capacities remain orders of magnitude below the empirical-potential regime~\cite{kavanagh2026fast}.
Pruning the message-passing depth of foundation models and partitioning the graph across GPUs accelerates deployment, yet the gap remains at orders of magnitude~\cite{kong2026scalable}.
Universality and scale are pursued together by NEP89, which extends the GPU-oriented neuroevolution potential to 89 elements~\cite{fan2021nep,liang2026nep89}, at markedly higher errors than the equivariant models.
Completing the universality--accuracy--speed triangle therefore requires a different architecture, not only a faster implementation.

\begin{figure}[t]
    \centering
    \includegraphics[width=\linewidth]{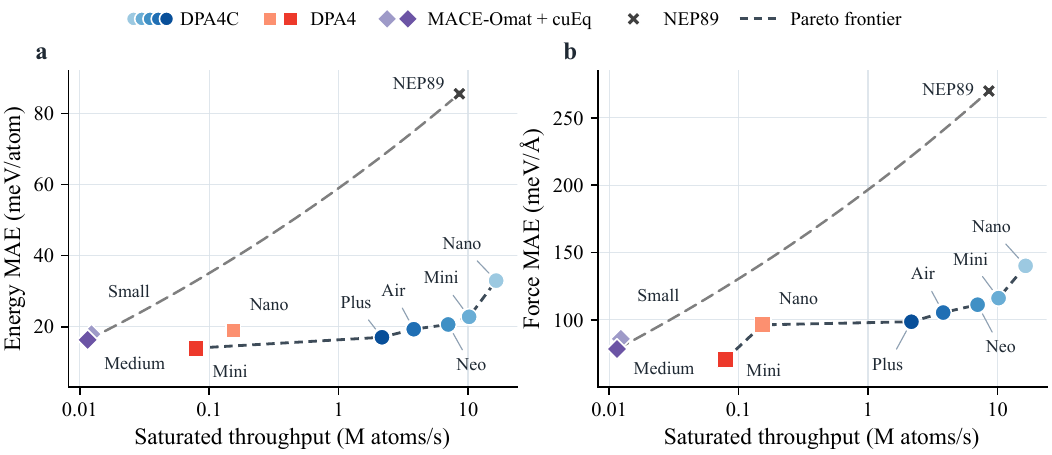}
    \caption{%
        Accuracy and saturated throughput on OMat24.
        \textbf{a}, Energy-per-atom MAE; \textbf{b}, force MAE on the OMat24 validation set, each plotted against saturated throughput on one NVIDIA H20.
        Lower error and higher throughput are preferred.
        Dark dashed lines mark the measured Pareto frontier; medium-gray dashed curves connect MACE-Omat-Medium and NEP89, the trade-off available before this work.
        DPA4C and NEP89 throughputs are calculated from complete NVT molecular-dynamics steps in LAMMPS/Kokkos and GPUMD, respectively, using the same initial diamond-carbon geometry at each system size.
        DPA4 and MACE-Omat throughputs are the largest mean rates found by a diamond-supercell size scan of ASE energy, force and stress evaluations, with MACE on the cuEquivariance-accelerated path~\cite{ase-paper,nvidia_cuequivariance,li2026dpa4}.
        DPA4 and MACE-Omat accuracy values are taken from the DPA4 study~\cite{li2026dpa4}, whereas the released NEP89 model was evaluated independently in this work~\cite{liang2026nep89}; its complete evaluation protocol is given in Supplementary Note~\ref{si:nep89-omat24-evaluation}.
        Because the backends differ, these comparisons summarize the stated engine-specific deployment measurements rather than hardware-independent model efficiency.
    }
    \label{fig:accuracy-throughput}
\end{figure}

Here we introduce DPA4C, a compact equivariant potential co-designed with its compressed CUDA operators, and demonstrate that one architecture can hold all three corners at deployment scale.
The design rests on two architectural constraints.
Every learned function evaluated on an edge depends only on the interatomic distance and the element pair, so at deployment the learned model collapses into an interpolation table and a finite cache~\cite{lu2022dp}.
Every learned state is local to one atom and is processed in fixed-size tiles, so the memory that grows with the system is dominated by the neighbor graph and the physical outputs.
In addition, DPA4C keeps a single message-passing layer that consumes only neighbor coordinates and types, so distributed runs exchange no learned features between GPUs.
The suffix C names the resulting pair of properties: compact and compressible by construction.
\FloatBarrier

We benchmark five variants, spanning a 49-fold range of parameter counts, on OMat24~\cite{barroso2024open}, MatPES~\cite{kaplan2025matpes} and OMol25~\cite{levine2025open}.
On OMat24, the five variants establish the high-throughput end of the measured accuracy--throughput Pareto frontier (Fig.~\ref{fig:accuracy-throughput}).
The largest variant, DPA4C-Plus, approaches the accuracy of the MACE-Omat models at about two orders of magnitude higher measured throughput.
The most compact variant, DPA4C-Nano, reduces the energy, force and stress errors of the independently evaluated NEP89 model by 61.4\%, 48.1\% and 34.3\%, respectively, while delivering about twice its saturated throughput on one NVIDIA H20.
Training each OMat24 variant requires 6.8--35.2 H20 GPU-hours.
The same compact, compressible model family extends to the charged and open-shell molecular systems of OMol25, broadening its chemical scope beyond materials.

Every variant completes multimillion-atom simulations on a single GPU, and on 1,024 16-GB NVIDIA V100 GPUs the five variants advance 2.048 billion atoms at 83.3--91.2\% weak-scaling efficiency.
Across complete single-GPU size scans on the same V100 hardware, Nano reaches 1.8 and 2.5 times the saturated throughput of MEAM for diamond carbon and FCC copper, respectively~\cite{baskes1992meam}, placing universal-potential MD within the empirical-potential range.
Together, these results show that one architecture now holds all three corners at once: universality across chemistry, near-first-principles accuracy and the speed of empirical potentials.

%% file: chap/results.tex
\section{Results}
\label{sec:results}

\subsection{The DPA4C architecture and compressed execution}
\label{sec:architecture}

\begin{figure}[!t]
    \centering
    \includegraphics[width=1.\textwidth]{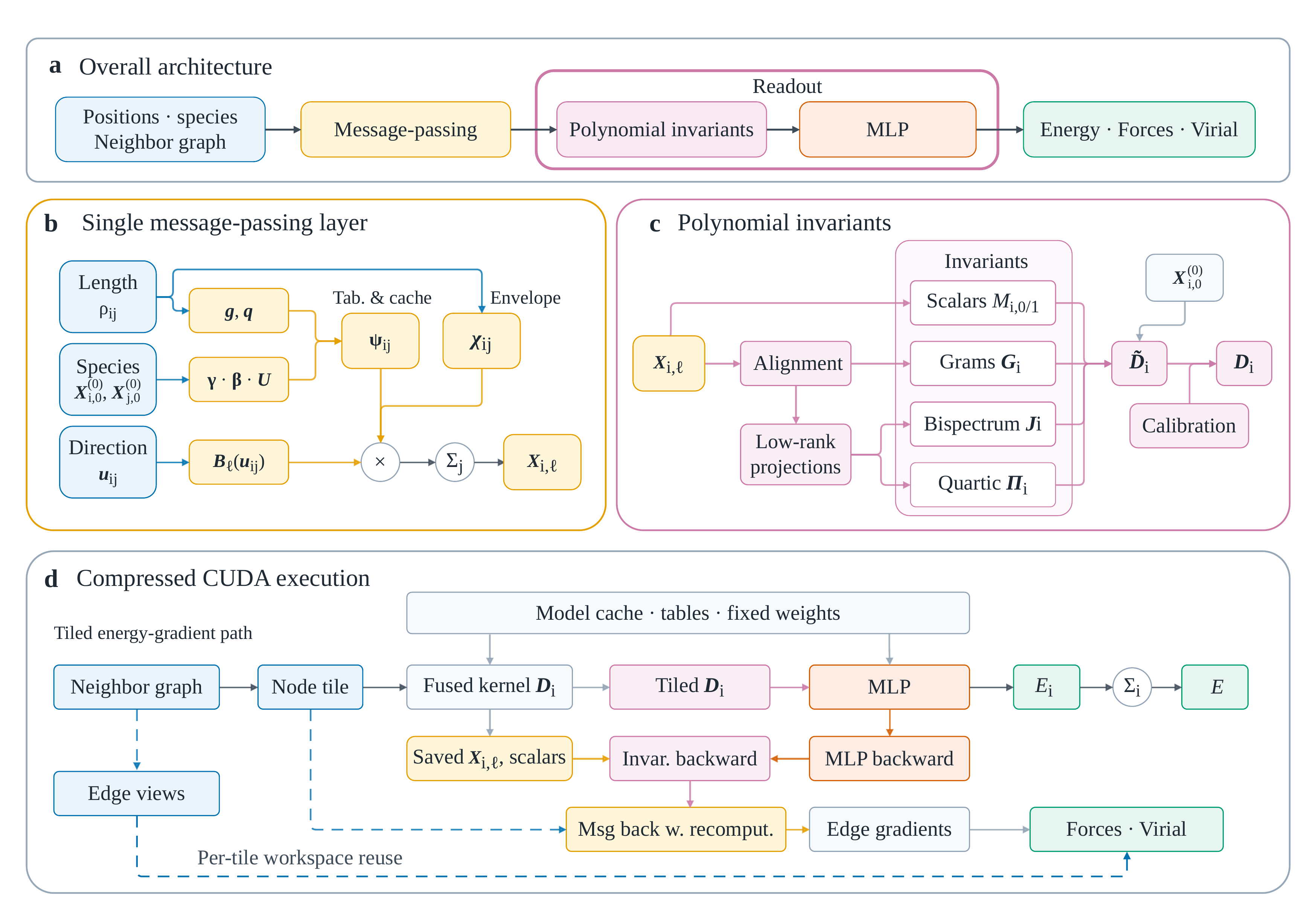}
    \caption{%
        DPA4C architecture and compressed CUDA execution.
        \textbf{a}, Overall architecture.
        \textbf{b}, The single message-passing layer: factorized edge messages and one aggregation over the neighborhood.
        \textbf{c}, Polynomial invariants and their fixed calibration (Methods~\S\ref{sec:methods-output}).
        \textbf{d}, Compressed CUDA execution: tiled forward--backward evaluation with message recomputation and force--virial assembly.
    }
    \label{fig:architecture}
\end{figure}

DPA4C is an equivariant graph network designed around the two constraints stated in the introduction.
From the atomic positions and species, a single message-passing layer builds equivariant features on each atom, and a nonlinear readout maps these features to the atomic energy (Fig.~\ref{fig:architecture}a).
Forces and the virial follow by differentiation.

Each atom enters the message-passing layer with an initial feature \(\bX^{(0)}_{i,0}\) encoding its species, and one aggregation over its neighborhood produces its updated equivariant features (Fig.~\ref{fig:architecture}b; Methods~\S\ref{sec:methods-moments}).
The message on an edge \(j\to i\) carries a learned amplitude \(\bpsi_{ij}\), which depends on the interatomic distance \(\rho_{ij}\) through a shared radial map and on the initial features of the two endpoints through cached modulation coefficients.
The amplitude is multiplied by the real Cartesian harmonics \(\bB_\ell(\bu_{ij})\) of the edge direction \(\bu_{ij}\) and by a smooth cutoff envelope \(\chi_{ij}\).
The degree-\(\ell\) harmonic block forms an irreducible representation of \(\Orth(3)\), for all degrees up to the maximum \(L\).
The aggregation accumulates all degrees and channels of the messages in one reduction over the neighborhood.
The same reduction accumulates the two sums that define the smooth neighborhood normalizers \(M_{i,0}\) and \(M_{i,1}\), which rescale the features degree by degree.
The resulting node features \(\bX_{i,\ell}\) transform equivariantly but remain local to atom \(i\).
No second layer exists to propagate them.

The readout proceeds in two steps.
Fixed polynomial contractions collect \(\Orth(3)\) invariants of the features into the fixed-length invariant feature vector \(\bD_i\) (Fig.~\ref{fig:architecture}c), and a residual multilayer perceptron (MLP) maps \(\bD_i\) to the atomic energy (Methods~\S\ref{sec:methods-output}).
The only learned operations in the first step are linear maps acting on the channel indices.
For each non-scalar degree the readout retains the complete channel Gram matrix \(\bG_{i,\ell}\), so no channel direction is discarded at quadratic order.
Learned degree-wise low-rank projections select the channel subspaces that enter the third-order bispectrum \(\bJ_i\) and the fourth-order projected quartic invariant \(\bPi_i\).
These invariants are concatenated with the scalar features, the two neighborhood normalizers and the initial feature \(\bX^{(0)}_{i,0}\) of the center.
A fixed componentwise calibration of the concatenation yields \(\bD_i\).
The requirement guiding this choice is that the invariants determine the node features up to a global rotation or reflection.
How far the retained set meets this requirement is stated precisely in Methods~\S\ref{sec:methods-readout}.
Because the MLP acts on invariants alone, the symmetry of the energy is exact by construction rather than learned.

The factorization of the message makes its learned content directly compressible (Fig.~\ref{fig:architecture}b; Methods~\S\ref{sec:methods-compression}; Supplementary Note~\ref{si:quintic}).
The radial maps depend on one scalar and are replaced by a quintic Hermite interpolation table.
The ordered-pair coefficients belong to a finite cache.
The envelope, the Cartesian harmonics, the aggregation and the polynomial invariants remain analytic.
Compression therefore changes how the learned distance function is evaluated, not the symmetry or the polynomial form of the invariant feature vector.
The number of shared radial modes \(R\) increases the edge arithmetic and the cache width but leaves both the flat feature width and the width of the invariant feature vector unchanged.

A single CUDA kernel carries each atom from its edge interval to its invariant feature vector \(\bD_i\), fusing the construction of the messages, their aggregation and the evaluation of the polynomial invariants, rather than executing them as a sequence of framework tensor operators (Fig.~\ref{fig:architecture}d; Supplementary Note~\ref{si:compressed-algorithm}).
The graph is first converted to a canonical form in which the edges are sorted by destination atom and indexed by a compressed sparse row (CSR) pointer, so each atom owns one contiguous edge interval.
One warp scans this interval, evaluates the radial table, the pair modulation, the envelope and the harmonics in registers, and accumulates the features and the two normalizer sums on the fly.
No message is ever written to global memory.
The same kernel then normalizes the accumulated features and evaluates the polynomial invariants, writing the feature vector \(\bD_i\) of each atom in the tile.
The framework tensor implementation must materialize the message of every edge, occupying memory proportional to the number of edges times the feature width.
By contrast, the fused kernel keeps only a per-destination working set in registers.

Forces require the gradient of the energy with respect to every edge displacement.
Storing the messages of the forward pass would make these gradients cheap to obtain, but would reintroduce a learned state on every edge.
DPA4C instead retains only the aggregated node features and the two neighborhood normalizers.
The backward pass differentiates the MLP and the polynomial invariants, then revisits each edge and recomputes its message while forming \(\partial E/\partial\br_{ij}\), spending arithmetic to save memory (Fig.~\ref{fig:architecture}d).
A separate kernel assembles the forces and the virial from these edge gradients, traversing the destination- and source-sorted edge lists so that each atom accumulates its contributions in a fixed order, without atomic floating-point additions (Methods~\S\ref{sec:methods-compression}).

What remains width dependent is the per-atom work of the readout: the invariant feature vector, the MLP activations and their gradients.
DPA4C processes atoms in contiguous node tiles of at most 131,072 atoms, completing the feature-vector evaluation, the MLP and its backward pass within one tile before the workspace is reused for the next (Fig.~\ref{fig:architecture}d).
Memory that depends on the model width is therefore bounded by the tile size, while only the graph and the physical outputs grow with the numbers of edges and atoms.
Changing the channel width \(C_0\), the maximum angular degree \(L\), the number of radial modes \(R\) or the MLP width changes the arithmetic within a tile, not any array that spans the system.

\subsection{OMat24 accuracy and throughput}
\label{sec:accuracy-efficiency}

\begin{table}[!t]
    \PaperTableStyle
    \caption{OMat24 accuracy and computational cost.
    Energy, force and stress are MAEs on the validation set; throughput is the number of atoms, in millions, advanced by one simulation step per second of wall time on a single H20, measured at system sizes large enough to saturate the GPU; training time is the total training cost in equivalent H20 GPU-hours.
    Lower errors, higher throughput and lower training cost are preferred.}
    \label{tab:omat24-accuracy}
    \begin{threeparttable}
        \begin{tabular*}{\linewidth}{@{\extracolsep{\fill}} l c c c c S[table-format=1.2e1,scientific-notation=true,retain-zero-exponent=true,round-mode=places,round-precision=2] c @{}}
            \toprule
            \textbf{Model}                     &
            \textbf{Energy}~$\downarrow$\tnote{a}           &
            \textbf{Force}~$\downarrow$\tnote{a}            &
            \textbf{Stress}~$\downarrow$\tnote{a}           &
            \textbf{Params}                    &
            \multicolumn{1}{c}{\textbf{Throughput}~$\uparrow$\tnote{a,b}} &
            \textbf{Train. time}~$\downarrow$\tnote{c}       \\
            \midrule
            NEP89\tnote{d}                       & 85.5 & 269.9 & 6.7 & 0.976M & 8.5674  & --      \\
            MACE-Omat-Small\tnote{e}             & 17.9 & 85.9  & 3.5 & 8.222M & 0.0123    & --      \\
            MACE-Omat-Medium\tnote{e}            & 16.3 & 78.4  & 3.3 & 9.063M & 0.0115    & --      \\
            EquiformerV3, $L_{\max}=4$~\cite{liao2026equiformerv3,li2026dpa4} & 10.4 & 43.5 & 2.6 & 30M & 0.00043 & 7,281.7 \\
            \addlinespace[0.25em]
            \midrule
            \addlinespace[0.15em]
            DPA4-Nano~\cite{li2026dpa4} & 18.9 & 96.3 & 3.4 & 0.480M & 0.153  & 80.3    \\
            DPA4-Mini  & 14.0 & 70.7 & 2.9 & 0.655M & 0.0792 & 153.2   \\
            DPA4-Pro   & 9.4  & 42.7 & 2.4 & 25.2M  & 0.00136 & 2,244.5 \\
            \addlinespace[0.25em]
            \midrule
            \multicolumn{7}{l}{\textit{DPA4C (this work)}} \\
            DPA4C-Nano & 33.0 & 140.1 & 4.4 & 0.030M & 16.483 & 6.8  \\
            DPA4C-Mini & 22.8 & 116.1 & 3.8 & 0.146M & 10.1919 & 9.6  \\
            DPA4C-Neo  & 20.6 & 111.3 & 3.7 & 0.342M & 7.0181  & 13.7 \\
            DPA4C-Air  & 19.3 & 105.4 & 3.6 & 0.434M & 3.8059  & 17.8 \\
            DPA4C-Plus & 17.0 & 98.6  & 3.4 & 1.457M & 2.1617  & 35.2 \\
            \bottomrule
        \end{tabular*}
        \begin{tablenotes}[flushleft]
            \PaperTableNotesStyle
            \item[a] Energy, force, stress and throughput are reported in meV/atom, meV/\Ang{}, meV/\Ang\(^{3}\) and M atoms/s, respectively.
            Energy is normalized by the number of atoms, while force and stress are componentwise MAEs.
            \item[b] DPA4C and NEP89 throughputs use complete molecular-dynamics steps in LAMMPS/Kokkos and GPUMD, respectively; DPA4, EquiformerV3 and MACE-Omat use ASE energy, force and stress evaluations~\cite{li2026dpa4}.
            \item[c] Train. time is reported in H20 GPU-hours; a dash indicates that no comparable total training cost is available.
            \item[d] NEP89 incorporates D3(BJ) interactions; consistent with its published data convention, PBE-D3(BJ) contributions were added to the OMat24 reference labels.
            Its parameter count is calculated from the released configuration using the expression of Liang \textit{et al.
            }~\cite{liang2026nep89}.
            The complete protocol is given in Supplementary Note~\ref{si:nep89-omat24-evaluation}.
            \item[e] MACE-OMAT-0 supplies the released checkpoints~\cite{batatia2025crosslearning}; their parameter counts are obtained directly from the checkpoints, and their OMat24 validation MAEs were obtained by the independent evaluation reported in the DPA4 study~\cite{li2026dpa4}.
        \end{tablenotes}
    \end{threeparttable}
\end{table}

We assessed the accuracy--throughput range of DPA4C on OMat24 by training five variants on the published training split and evaluating them on the validation set.
Table~\ref{tab:omat24-accuracy} reports their MAEs, parameter counts, saturated single-H20 throughputs and training costs alongside the reference models.
The variants span 29,809 to 1,456,961 trainable parameters through coordinated changes in channel width, angular degree, radial modes and MLP width.
Complete configurations are given in Supplementary Tables~\ref{si:tab-model-variants} and~\ref{si:tab-omat24-training}.
Among the references, the released MACE-Omat checkpoints represent the widely used equivariant universal potentials~\cite{batatia2025crosslearning}, and the three DPA4 models define the accuracy--throughput frontier reached by equivariant message-passing architectures~\cite{li2026dpa4}, with EquiformerV3 as an independently developed model at their accuracy end~\cite{liao2026equiformerv3}.
NEP89 is the fastest universal MLIP available before this work~\cite{liang2026nep89}.

The five variants establish the high-throughput end of the measured accuracy--throughput Pareto frontier.
On both the energy and the force panel of Fig.~\ref{fig:accuracy-throughput}, every variant lies on the frontier, the DPA4 models hold its accuracy end, and neither NEP89 nor MACE-Omat lies on it.
From Nano to Plus, all three MAEs decrease monotonically while the saturated throughput falls from 16.48 to 2.16~M atoms/s (Table~\ref{tab:omat24-accuracy}).

At nearly matched parameter count, DPA4C-Air keeps all three MAEs within 10\% of DPA4-Nano while delivering 24.9 times its measured throughput at 77.8\% lower reported training cost.
DPA4C-Plus closes the remaining accuracy gap.
Relative to DPA4-Nano, its energy MAE is about 10\% lower, its force MAE is 2.4\% higher and its stress MAE is the same at the reported precision.
Plus reaches this accuracy at 14.1 times the measured DPA4-Nano throughput, and its reported training cost is less than half.
Relative to NEP89, DPA4C-Nano is simultaneously more accurate and faster.
With 3.1\% as many trainable parameters, it lowers the energy, force and stress MAEs by 61.4\%, 48.1\% and 34.3\%, respectively, at 1.92 times the throughput.
DPA4C-Mini lowers the three MAEs further, by 73.3\%, 57.0\% and 43.3\%, while remaining 19\% faster.

\subsection{MatPES R2SCAN materials benchmark}
\label{sec:matpes-accuracy}

\begin{table}[!htbp]
    \PaperTableStyle
    \caption{MatPES accuracy and training cost, on the R2SCAN-2025.2 test split.
    Energy, force and stress are MAEs; training time is the total training cost in equivalent H20 GPU-hours.
    Lower errors and lower training cost are preferred.}
    \label{tab:matpes-accuracy}
    \begin{threeparttable}
        \begin{tabular*}{\linewidth}{@{\extracolsep{\fill}} l c c c c c @{}}
            \toprule
            \textbf{Model}                     &
            \textbf{Energy}~$\downarrow$\tnote{a}           &
            \textbf{Force}~$\downarrow$\tnote{a}            &
            \textbf{Stress}~$\downarrow$\tnote{a}           &
            \textbf{Params}                    &
            \textbf{Train. time}~$\downarrow$\tnote{b}       \\
            \midrule
            DPA4-Nano~\cite{li2026dpa4} & 30.0                    & 142.5                    & 5.2                    & 0.480M & 4.6 \\
            DPA4-Mini  & 20.7                    & 108.5                    & 3.7                    & 0.655M & 9.2 \\
            \addlinespace[0.25em]
            \midrule
            \multicolumn{6}{l}{\textit{DPA4C (this work)}} \\
            DPA4C-Nano & 51.4                    & 193.6                    & 7.0                    & 0.030M & 1.5 \\
            DPA4C-Mini & 33.9                    & 164.4                    & 5.6                    & 0.146M & 1.5 \\
            DPA4C-Neo  & 30.1                    & 159.7                    & 5.3                    & 0.342M & 1.6 \\
            DPA4C-Air  & 27.7                    & 153.3                    & 5.0                    & 0.434M & 1.9 \\
            DPA4C-Plus & 24.5                    & 152.9                    & 4.7                    & 1.457M & 2.8 \\
            \bottomrule
        \end{tabular*}
        \begin{tablenotes}[flushleft]
            \PaperTableNotesStyle
            \item[a] Values are MAEs in meV/atom for energy, meV/\Ang{} for force and meV/\Ang\(^{3}\) for stress; energy is normalized by the number of atoms, while force and stress are componentwise MAEs.
            \item[b] Train. time is reported in H20 GPU-hours for the complete protocol used by each run; the training lengths are given in Supplementary Table~\ref{si:tab-matpes-training}.
        \end{tablenotes}
    \end{threeparttable}
\end{table}

We evaluated the same five DPA4C variants on the MatPES R2SCAN-2025.2 test split, probing the model family at a different density-functional level and on a training set more than two orders of magnitude smaller than OMat24~\cite{kaplan2025matpes}.
Complete training configurations are given in Supplementary Table~\ref{si:tab-matpes-training}.
Table~\ref{tab:matpes-accuracy} compares the resulting MAEs with the DPA4 references.

The five variants retain their OMat24 ordering, with all three MAEs decreasing monotonically from Nano to Plus.
As on OMat24, DPA4C-Plus improves on DPA4-Nano in energy but not in force.
Its energy MAE is 18.3\% lower and its stress MAE is also lower, its force MAE is 7.3\% higher, and training uses 39.1\% fewer reported H20 GPU-hours.
At the efficiency end of the family, DPA4C-Nano delivers about two orders of magnitude higher measured throughput than DPA4-Nano (Fig.~\ref{fig:accuracy-throughput}).
Its force and stress MAEs exceed the DPA4-Nano values by about 36\%, and its energy MAE by about 71\%.
DPA4-Mini retains the lowest errors overall and defines the accuracy-oriented end of this comparison.

\subsection{OMol25 molecular benchmark}
\label{sec:omol25-accuracy}

\begin{table}[!t]
    \PaperTableStyle
    \caption{OMol25 accuracy and training cost, on the OMol-0 out-of-distribution composition validation split.
    Energy and force are MAEs; training time is the total training cost in equivalent H20 GPU-hours.
    Lower errors and lower training cost are preferred.}
    \label{tab:omol25-accuracy}
    \begin{threeparttable}
        \begin{tabular*}{\linewidth}{@{\extracolsep{\fill}} l c c c c @{}}
            \toprule
            \textbf{Model}                         &
            \textbf{Energy}~$\downarrow$\tnote{a}     &
            \textbf{Force}~$\downarrow$\tnote{a}      &
            \textbf{Params}                        &
            \textbf{Train. time}~$\downarrow$\tnote{b} \\
            \midrule
            eSEN-sm-cons.~\cite{fu2025esen,levine2025open}       & 1.77  & 0.190 & 6.3M        & --    \\
            MACE-OMol-L-0~\cite{batatia2022mace,levine2025open} & 4.56  & 0.250 & --          & --    \\
            \addlinespace[0.25em]
            \midrule
            DPA4-Nano~\cite{li2026dpa4} & 8.02 & 0.776 & 0.480M & 283.2 \\
            DPA4-Mini                         & 4.97 & 0.502 & 0.655M & 656.3 \\
            \addlinespace[0.25em]
            \midrule
            \multicolumn{5}{l}{\textit{DPA4C (this work)}} \\
            DPA4C-Nano & 40.97 & 2.485 & 0.035M & 33.6  \\
            DPA4C-Mini & 28.58 & 1.881 & 0.200M & 56.7  \\
            DPA4C-Neo  & 24.74 & 1.731 & 0.539M & 91.2  \\
            DPA4C-Air  & 21.37 & 1.597 & 0.630M & 142.1 \\
            DPA4C-Plus & 16.77 & 1.383 & 2.189M & 340.9 \\
            \bottomrule
        \end{tabular*}
        \begin{tablenotes}[flushleft]
            \PaperTableNotesStyle
            \item[a] Values are MAEs in kcal/mol for total energy and kcal/mol/\Ang{} for force. Values reported in meV and meV/\Ang{} are converted using $1~\mathrm{meV}=0.0230605~\mathrm{kcal/mol}$ and $1~\mathrm{meV}/\Ang{}=0.0230605~\mathrm{kcal/mol}/\Ang{}$, respectively.
            \item[b] Train. time is reported in H20 GPU-hours; a dash indicates that no comparable total training cost is available.
        \end{tablenotes}
    \end{threeparttable}
\end{table}

We evaluated all five DPA4C variants on the OMol25 OMol-0 out-of-distribution composition validation split, providing a molecular counterpart to the two materials benchmarks~\cite{levine2025open}.
Complete training configurations are given in Supplementary Table~\ref{si:tab-omol25-training}.

On the molecular benchmark the accuracy gap to the DPA4 models widens substantially relative to the two materials benchmarks.
Scaling from Nano to Plus still lowers the total-energy and force MAEs monotonically, by 59.1\% and 44.3\% overall.
At 340.9 H20 GPU-hours, DPA4C-Plus uses 20.4\% more reported training compute than DPA4-Nano yet has a 109.1\% higher total-energy MAE and a 78.2\% higher force MAE.
The other OMol25 baselines reported in the same units lie further ahead.
The inference-throughput advantage measured on the materials benchmark (Fig.~\ref{fig:accuracy-throughput}) applies equally to molecular systems, because the inference cost of either architecture depends on the numbers of atoms and neighbors rather than on the elements or the chemical domain.

\subsection{Compressed deployment}
\label{sec:compression-fidelity}

With accuracy established, we turn to deployment.
Compressed execution preserves accuracy at the reported precision while bounding width-dependent memory and reducing the fixed costs of a complete molecular-dynamics step.

To test the complete deployed path, we exported a compressed and an uncompressed model from the same checkpoint for each OMat24 variant and evaluated both on all 1,074,643 validation structures.
At the 0.002~\Ang{} spacing used for deployment, the largest compressed--uncompressed MAEs among the five variants were \(3.52\times10^{-4}~\mathrm{meV/atom}\) for energy, \(1.18\times10^{-3}~\mathrm{meV/\Ang}\) for force and \(3.13\times10^{-5}~\mathrm{meV/\Ang^3}\) for stress.
A sweep of four spacings spanning the hundredfold range from 0.0005 to 0.05~\Ang{} further bounds the numerical effect (Supplementary Table~\ref{si:tab-compression-fidelity}).
The discrepancies remain similarly small up to 0.01~\Ang{}, and a clear increase in force and stress RMSE appears only at 0.05~\Ang{}.
Even across this sweep, every MAE against the reference labels remains unchanged at the reported precision.

A paired ablation isolates the effect of node tiling.
With the model, graph, precision and MD input fixed, tiling raises the largest completed system by 27.3\% for Nano and up to 155.0\% for Plus, while changing throughput by only \(-2.0\%\) to \(+0.3\%\) (Supplementary Table~\ref{si:tab-node-tiling}).
The memory benefit therefore grows with model width without an appreciable throughput penalty.

At Nano's speed the model is no longer the only bottleneck.
The fixed costs of the MD step itself, namely graph construction and the energy reduction, occupy a visible fraction of the step time.
Rebuilding these two operators saves 3.43~ms per step, 5.4\% of the complete Nano step and 0.7--3.4\% for the four larger variants (Supplementary Table~\ref{si:tab-whole-step-ablation}).

\subsection{Single-GPU performance benchmarks}
\label{sec:inference-scaling}

\begin{figure}[t!]
    \centering
    \includegraphics[width=\linewidth]{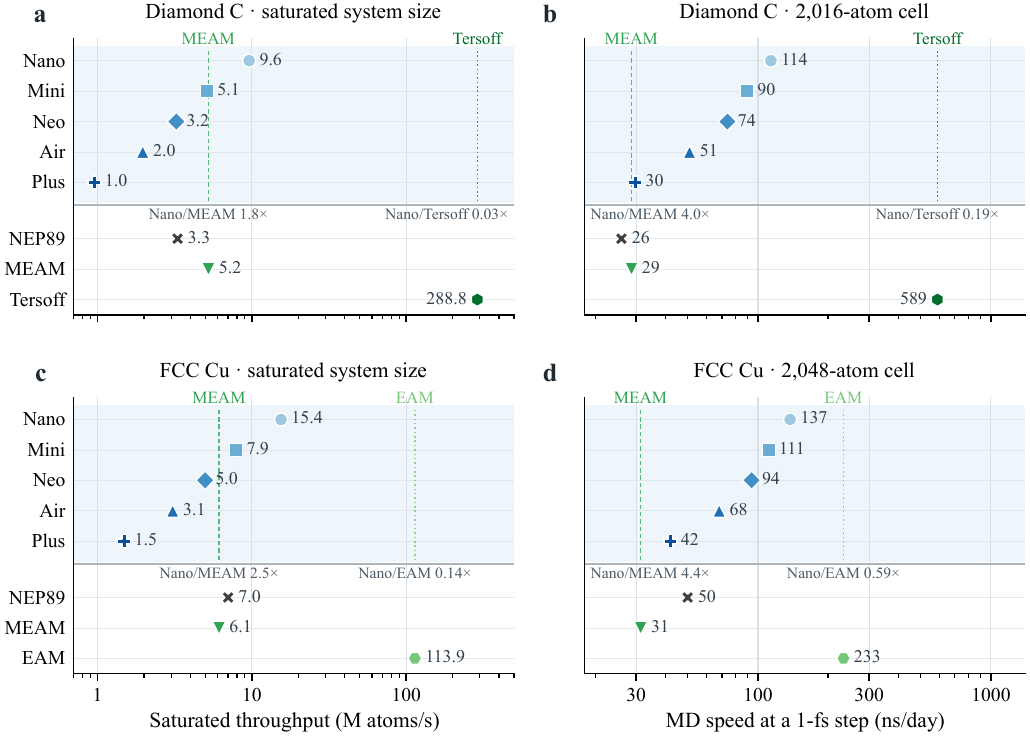}
    \caption{%
        Single-GPU performance of DPA4C, NEP89 and the empirical potentials on a 16-GB NVIDIA Tesla V100-SXM2 GPU.
        Rows show diamond carbon (\textbf{a},\textbf{b}; 158 neighbors per atom at the 6-\Ang{} learned-model cutoff) and FCC copper (\textbf{c},\textbf{d}; 78 neighbors per atom).
        \textbf{a},\textbf{c}, Saturated throughput: the largest mean throughput along each model's system-size scan, where each point averages three independent runs.
        \textbf{b},\textbf{d}, MD speed of a 2,016-atom (carbon) or 2,048-atom (copper) cell at a 1-fs time step.
        The shaded band holds the five DPA4C variants, dashed vertical lines carry the empirical-potential values into the band, and the annotations report the ratios of DPA4C-Nano to the empirical potentials.
        DPA4C and the empirical potentials run in LAMMPS/Kokkos, NEP89 in its native GPUMD; each empirical potential retains its native cutoff.
        The complete throughput curves are shown in Supplementary Fig.~\ref{si:fig-v100-crystal-scan}.
    }
    \label{fig:empirical-reference}
\end{figure}

We benchmarked the performance of the five DPA4C variants, NEP89~\cite{liang2026nep89} and a set of empirical potentials in diamond-carbon and FCC-copper crystals on one 16-GB NVIDIA Tesla V100-SXM2 GPU (Fig.~\ref{fig:empirical-reference}).
Every measurement times complete NVT molecular-dynamics steps and reports throughput: the number of atoms advanced by one step per second of wall time.
The empirical potentials are MEAM~\cite{baskes1992meam} and Tersoff~\cite{tersoff1989multicomponent} for carbon, and MEAM and the Mishin EAM potential~\cite{mishin2001copper} for copper.
These element-specific potentials serve as computational references rather than accuracy-matched baselines for the multi-element DPA4C models.
We scanned the system size from about one hundred atoms until each model ran out of GPU memory, and every throughput curve has the same shape (Supplementary Fig.~\ref{si:fig-v100-crystal-scan}; Supplementary Tables~\ref{si:tab-v100-carbon-scan} and~\ref{si:tab-v100-copper-scan}).
At small sizes the wall time per step stays near a fixed, implementation-specific floor, so throughput grows with atom count.
Once the system is large enough to fill the GPU, the time per step grows in proportion to the atom count and the throughput saturates.
The two limits of this curve correspond to the two ways molecular dynamics is deployed.
The saturated throughput sets the cost of simulating large systems, as in segregation at grain-boundary networks or reactive interfaces that require millions of atoms.
The speed of a cell of about two thousand atoms, expressed in simulated nanoseconds per day, sets the trajectory length affordable for a small system, as in studies of nucleation, defect kinetics or slow structural relaxation.

In the saturated regime, DPA4C overlaps, and in several cases exceeds, the throughput of MEAM and NEP89, while the fastest empirical potential in each crystal stays ahead (Fig.~\ref{fig:empirical-reference}a,c).
In diamond carbon, DPA4C-Nano reaches 1.84 times the saturated throughput of MEAM and 2.91 times that of NEP89.
In FCC copper, DPA4C-Nano reaches 2.52 times the MEAM throughput and 2.20 times that of NEP89, and DPA4C-Mini also exceeds both.
Tersoff remains 30.1 times faster than Nano in carbon, and EAM 7.38 times faster in copper.
The lower local coordination of copper at the common \(6~\Ang\) learned-model cutoff, 78 neighbors per atom versus 158 in carbon, reduces the work per atom and raises the throughput of every DPA4C variant relative to carbon.

In the small-cell regime the wall time per step approaches each implementation's fixed floor, so speed is set by per-step latency rather than by work per atom (Fig.~\ref{fig:empirical-reference}b,d).
This reverses part of the saturated ranking.
In diamond carbon every DPA4C variant advances more nanoseconds per day than MEAM and NEP89, and in FCC copper all five exceed MEAM while all but Plus exceed NEP89.
Nano reaches 4.0 times the MEAM speed in carbon and 4.4 times in copper, up from 1.84 and 2.52 at saturation.
The gap to Tersoff and EAM narrows to 0.19 and 0.59 times, from 0.03 and 0.14.
The compressed DPA4C step issues a short, fixed sequence of kernels (Methods~\S\ref{sec:methods-compression}), which keeps its floor low.
The floors of Tersoff and EAM remain lower still, but by a far smaller margin than their advantage in work per atom, and the gap narrows accordingly.

Memory capacity orders the models differently (Supplementary Tables~\ref{si:tab-v100-carbon-scan} and~\ref{si:tab-v100-copper-scan}).
On the 16-GB V100, the five DPA4C variants complete 2.10 million atoms in diamond carbon.
In FCC copper, Nano through Air complete 4.20 million atoms, whereas the Plus capacity boundary falls below this size, with 3.92 million atoms completed.
The tiled execution path keeps the width-dependent workspace fixed per tile rather than proportional to the atom count (Supplementary Eq.~\eqref{si:eq-tile-memory}).
This fixed term is largest for Plus, the widest variant, and on a 16-GB device it is what separates Plus from the other four.
NEP89 completes 1.97 million atoms in both crystals.
MEAM completes 5.27 million atoms in carbon and 8.39 million in copper, and Tersoff and EAM each complete about 29.4 million.
From carbon to copper the DPA4C capacity roughly doubles as the neighbor count halves.

\subsection{Distributed scaling to 1,024 GPUs}
\label{sec:multigpu-scaling}

\begin{figure}[t!]
    \centering
    \includegraphics[width=\linewidth]{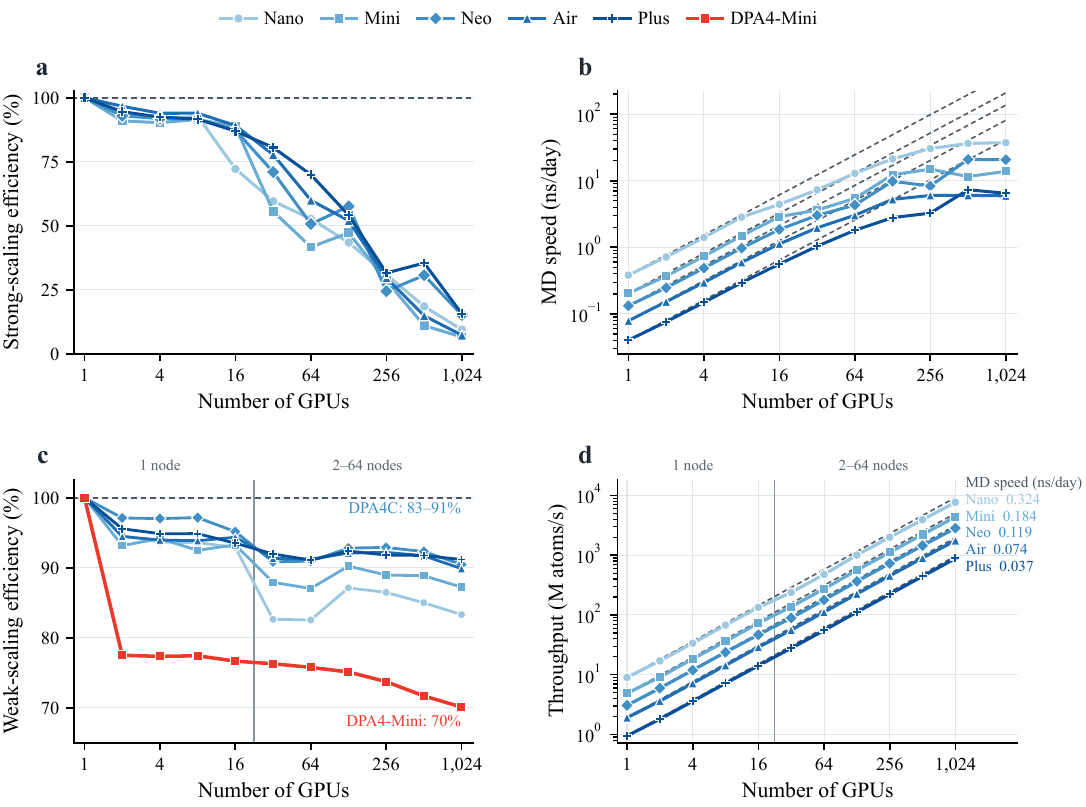}
    \caption{%
        Strong and weak scaling of complete NVT molecular-dynamics steps on 16-GB NVIDIA Tesla V100-SXM2 GPUs.
        \textbf{a}, Strong-scaling efficiency of the five DPA4C variants for a fixed 2,000,376-atom system.
        \textbf{b}, Corresponding measured MD speed for the same fixed system; thin dark-grey dashed lines give ideal linear scaling from each variant's one-GPU MD speed.
        \textbf{c}, Weak-scaling efficiency for the five DPA4C variants at 2,000,376 atoms per GPU and DPA4-Mini at 5,832 atoms per GPU, each normalized independently to the corresponding one-GPU rate.
        \textbf{d}, Aggregate DPA4C weak-scaling throughput at 2,000,376 atoms per GPU; thin dark-grey dashed lines give ideal linear scaling from each variant's one-GPU throughput, and endpoint labels identify the variant and measured MD speed for 2,048,385,024 atoms on 1,024 GPUs.
        In \textbf{c} and \textbf{d}, the vertical line separates single-node allocations, up to 16 GPUs, from multi-node allocations of 32 GPUs and beyond.
        Lines show means of ten independent allocations in \textbf{a,b} and medians of three in \textbf{c,d}; the full repeat ranges and supporting timing fractions are reported in Supplementary Figs.~\ref{si:fig-v100-strong-detail} and~\ref{si:fig-v100-comparison-detail}.
    }
    \label{fig:multigpu-scaling}
\end{figure}

We measured DPA4C on 1--1,024 16-GB NVIDIA Tesla V100-SXM2 GPUs, covering 64 nodes (Fig.~\ref{fig:multigpu-scaling}).
Strong scaling fixes the same cubic 2,000,376-atom system for all five variants and follows efficiency and measured MD speed with increasing GPU count (Fig.~\ref{fig:multigpu-scaling}a,b).
At eight GPUs, the variants retain 91.5--93.9\% efficiency.
Because Nano does the least work per atom, its local computation is the first to become too small to hide communication and other fixed step costs.
At 16 GPUs it retains 72.2\%, whereas Mini through Plus retain 87.0--89.1\%.
At 1,024 GPUs, each process owns about 1,953 atoms.
The efficiencies are 6.6--15.7\% and the measured MD speeds are 5.93--37.22~ns/day.
Doubling the allocation from 512 to 1,024 GPUs changes the MD speed by only \(-12\%\) to \(+19\%\) across the five variants, marking the throughput plateau at low local work.

Weak scaling fixes a cubic local domain of 2,000,376 atoms per GPU from one to 1,024 GPUs.
The 1,024-GPU endpoint contains 2,048,385,024 atoms and retains 83.3--91.2\% of the single-GPU throughput across the five variants (Fig.~\ref{fig:multigpu-scaling}c).
Nano reaches 7.67 billion atoms/s and Plus reaches 0.881 billion atoms/s, with Mini, Neo and Air between them (Fig.~\ref{fig:multigpu-scaling}d).
At a 1-fs time step, the corresponding MD speeds are 0.324~ns/day for Nano and 0.037~ns/day for Plus.
The transition from 16 GPUs on one node to 32 GPUs on two nodes lowers Nano's weak-scaling efficiency from 92.9\% to 82.6\%.
The corresponding decreases are 1.5--5.3 percentage points for the four larger variants.

Because the single message-passing layer needs only the neighbor coordinates and types, DPA4C exchanges ghost-atom coordinates and types before evaluation and returns ghost force and virial contributions afterwards, carrying no learned feature state between MPI domains.
DPA4-Mini, in contrast, additionally communicates learned ghost features during message passing and serves as the deployment reference at its own feasible workloads.
With 5,832 atoms per GPU, it retains 70.1\% weak-scaling efficiency on 1,024 GPUs (Fig.~\ref{fig:multigpu-scaling}c).
Its fixed 8,000-atom strong-scaling efficiency falls to 57.3\% on eight GPUs and 34.6\% on 32 GPUs, where each GPU receives 250 atoms (Supplementary Fig.~\ref{si:fig-v100-comparison-detail}a).
Each curve reports the fraction of its own single-GPU throughput.
Because the workloads, runtime paths and intermediate storage all differ between DPA4C and DPA4-Mini, the scaling gap does not isolate the cost of learned-feature exchange.

%% file: chap/discussion.tex
\section{Discussion}
\label{sec:discussion}

DPA4C holds the three corners of the universality--accuracy--speed triangle at deployment scale.
Universality rests on pre-training over large-scale datasets.
The same five variants span the chemistry of the materials datasets OMat24 and MatPES and of the molecular dataset OMol25.
Accuracy and speed are demonstrated jointly.
The five variants establish the high-throughput end of the measured accuracy--throughput Pareto frontier.
DPA4C-Plus approaches the accuracy of the MACE-Omat checkpoints at about two orders of magnitude higher measured throughput, and DPA4C-Nano lowers the energy, force and stress MAEs of the fastest existing universal model by 34.3--61.4\% at about twice its throughput.
On deployment hardware, Nano exceeds the saturated MEAM throughput in both the diamond-carbon and FCC-copper benchmarks, and in small cells every DPA4C variant runs faster than MEAM.
On 1,024 16-GB NVIDIA V100 GPUs, the five variants scale molecular dynamics to 2.048 billion atoms at 83.3--91.2\% weak-scaling efficiency.
Universal-potential MD therefore now operates in the empirical-potential regime at higher accuracy than the fastest existing universal model.

The deployment performance traces to the co-design of the architecture with its operators.
Every learned edge function depends only on the interatomic distance and the element pair, so deployment collapses the learned model into an interpolation table and a finite cache that reproduce the uncompressed model without loss of accuracy.
Every learned state is local to one atom and is processed in fixed-size tiles, so no learned array spans the simulated system and the memory that depends on the model width is a fixed per-tile workspace.
The single message-passing layer consumes only neighbor coordinates and types, leaving distributed runs to exchange geometry and physical outputs but no learned features.
Software optimization of a fixed architecture accelerates its execution but does not change what the model stores and communicates.
The co-design changes exactly that.

Several boundaries remain.
DPA4C is strictly local and contains no explicit electrostatics, dispersion, charge equilibration or response to external fields, so systems governed by long-range interactions lie outside its present scope.
Its radial functions are learned without a built-in short-range repulsion, so configurations well inside the shortest training distances are not guaranteed to be physical.
Adding long-range physical terms and distilling DPA4 into DPA4C are natural next steps.
Segregation at grain-boundary networks, multi-principal-element alloys and reactive interfaces are governed by short-range bonding yet require broad chemistry at empirical-potential size and speed, exactly the combination that DPA4C delivers.

%% file: chap/methods.tex
\section{Methods}
\label{sec:methods}

\subsection{Overall structure and constraints}
\label{sec:methods-setting}

A configuration of \(N\) atoms is specified by Cartesian positions \(\br_1,\dots,\br_N\), measured in \Ang, and by chemical species \(a_1,\dots,a_N\) drawn from a fixed set of \(T\) element types.
DPA4C adopts the locality ansatz shared by neural-network and Deep Potential interatomic potentials.
The potential energy, in eV, is a sum of atomic contributions,
\begin{equation}
    E=\sum_{i=1}^{N}E_i.
    \label{eq:energy-decomposition}
\end{equation}
Here \(E_i\) depends only on the atoms inside a sphere of radius \(\rcut\) centered on atom \(i\)~\cite{behler2007generalized,zhang2018deep}.
Each atomic contribution is produced by two components.
A single message-passing layer aggregates the neighborhood into equivariant node features \(\bX_{i,\ell}\), one block for each angular degree \(\ell\) (\S\ref{sec:methods-moments}).
A nonlinear readout then maps these features to the atomic energy in two steps.
A finite set of polynomial invariants of the features is collected into the invariant feature vector \(\bD_i\), and a multilayer perceptron (MLP) maps \(\bD_i\) to the scalar \(E_i\) (\S\ref{sec:methods-readout}).

Three constraints shape this construction.
The first is symmetry.
The energy is unchanged by a rigid translation, by a relabeling of identical atoms and by any element of the orthogonal group in three dimensions, \(\Orth(3)=\bigl\{\bR\in\R^{3\times3}:\bR\tp\bR=\bI\bigr\}\), which contains the proper rotations (\(\det\bR=+1\), forming the subgroup \(\SO(3)\)) together with every rotation composed with a reflection (\(\det\bR=-1\)).
A quantity is \(\Orth(3)\)-invariant when replacing every \(\br_k\) by \(\bR\br_k\) leaves it unchanged, and \(\Orth(3)\)-equivariant when it instead transforms under a fixed linear representation of the group.
In DPA4C the symmetry is enforced by the structure itself.
The node features transform equivariantly, the first step of the readout produces invariants by construction, and the MLP acts on invariants alone, so no stage needs to learn or approximate the symmetry.
Invariance alone does not guarantee that distinct environments receive distinct feature vectors, and Section~\ref{sec:methods-readout} states what the resulting invariant set does and does not separate.
This constraint is shared by every symmetry-respecting potential.
The remaining two are the architectural constraints stated in the introduction, and they are what distinguish DPA4C within the equivariant family.

The second is compressibility.
The learned content of a message is a finite sum of terms, each multiplying a scalar function of the interatomic distance by a coefficient of the ordered species pair.
Compression follows directly.
The one-dimensional radial functions are replaced by one interpolation table shared by all species pairs, and the finitely many species coefficients by a cache.
This form constrains the design from the outset.
One radial map is shared across all angular degrees, so a single interpolation table serves every degree.
The cutoff envelope multiplies the completed amplitude rather than the radial basis, so it remains analytic after compression and its degree-dependent powers stay out of the table.
Section~\ref{sec:methods-compression} makes the separation explicit.

The third is compactness of the per-atom state.
The width of the node features \(\bX_{i,\ell}\), and every other width in the model, derives from three integer controls: the scalar channel width \(C_0\in\{8,16,32,64,128\}\), the maximum angular degree \(L\in\{2,3,4\}\) and the number of shared radial modes \(R\ge0\).

The models studied here are further delimited by the locality ansatz and by three modeling choices.
They are strictly local at a single cutoff and therefore carry no explicit long-range electrostatics or dispersion.
They use no external-field input.
The total charge and spin multiplicity enter only the OMol25 models, through additional trainable embeddings (Supplementary Table~\ref{si:tab-model-variants}).
Their radial functions are learned over the whole interval \([0,\rcut]\) without an imposed repulsive core, and every quantity passed to the MLP is invariant, so tensorial properties such as dipoles or polarizabilities would require a separate equivariant head.

\subsection{One message-passing layer}
\label{sec:methods-moments}
\label{sec:methods-amplitude}
\label{sec:methods-graph}

This section defines the single message-passing layer of DPA4C: the graph it acts on, the initial node features, the message carried by each edge and the aggregation that turns the messages into equivariant node features.

\paragraph{The neighbor graph.}
The input to the model is a directed neighbor graph whose nodes are atoms.
A directed edge \((i,j)\) exists whenever atom \(j\) lies inside the cutoff sphere of the center atom \(i\),
\begin{equation}
    \Nb_i=\bigl\{j:j\ne i,\;\lVert\br_j-\br_i\rVert<\rcut\bigr\},
    \label{eq:neighbor-set}
\end{equation}
where \(j\) labels a neighbor instance and \(\br_j\) is its Cartesian position.
In a periodic cell, distinct periodic images are distinct neighbor instances, and only the center atom in its own image is excluded.
Both orientations of a physical pair are carried, because each center accumulates its own node features and because the ordered pair of species \((a_i,a_j)\) conditions the edge.
No neighbor capacity is imposed.
The graph carries every neighbor inside the cutoff, so the model has no maximum coordination number.
The total number of directed edges is \(N_{\mathrm{edge}}=\sum_i\lvert\Nb_i\rvert\).

\paragraph{Edge geometry.}
Geometry enters only through the edge displacements \(\br_{ij}=\br_j-\br_i\), which makes translation invariance exact.
Each displacement is converted into a regularized length and direction,
\begin{equation}
    \rho_{ij}=\sqrt{\lVert\br_{ij}\rVert^{2}+\varepsilon^{2}},
    \qquad
    \bu_{ij}=\frac{\br_{ij}}{\rho_{ij}},
    \qquad
    \varepsilon=10^{-7}~\Ang.
    \label{eq:edge-geometry}
\end{equation}
The positive constant \(\varepsilon\) keeps \(\bu_{ij}\) and its derivatives finite for a coincident pair, at the price of a direction Jacobian of order \(1/\varepsilon\) there.
At ordinary atomistic separations, for which \(\lVert\br_{ij}\rVert\gg\varepsilon\), the pair \((\rho_{ij},\bu_{ij})\) agrees with the ordinary distance and unit direction to numerical precision.
The squared norm \(\lVert\bu_{ij}\rVert^{2}=1-\varepsilon^{2}/\rho_{ij}^{2}\) departs from unity by less than the single-precision resolution.

\paragraph{Smooth cutoff envelope.}
A cutoff envelope removes the discontinuity that would otherwise appear when an atom crosses the cutoff sphere.
DPA4C uses the polynomial envelope of DPA4~\cite{li2026dpa4}, itself of the family introduced for smooth message-passing potentials~\cite{gasteiger2019directional}, with its exponent fixed at five.
For a regularized length \(\rho\), define \(t=\clip(1-\rho/\rcut,0,1)\) and \(x=1-t\).
The envelope is
\begin{equation}
    \chi(\rho)=t^{4}\sum_{k=0}^{4}\binom{k+3}{3}x^{k}
    =t^{4}\bigl(1+4x+10x^{2}+20x^{3}+35x^{4}\bigr),
    \qquad
    \chi_{ij}=\chi(\rho_{ij}).
    \label{eq:envelope}
\end{equation}
The envelope satisfies \(\chi(0)=1\) and \(\chi(\rho)=0\) for \(\rho\ge\rcut\).
The factor \(t^{4}\) makes \(\chi\) three times continuously differentiable at the cutoff and its fourth derivative discontinuous there.
Because every edge contribution below carries at least one factor of \(\chi_{ij}\), an atom may cross the cutoff sphere without introducing a discontinuity in the energy, the forces or the force derivatives with respect to atomic coordinates.
An edge exactly at \(\rcut\) carries zero weight, so the boundary convention in Eq.~\eqref{eq:neighbor-set} is immaterial.

\paragraph{Initial node features.}
Each atom enters the network with equivariant node features supported on degree zero alone,
\begin{equation}
    \bX^{(0)}_{i,0}\in\R^{C_0},
    \qquad
    \bX^{(0)}_{i,\ell}=\bm{0}\quad(1\le\ell\le L),
    \label{eq:init-features}
\end{equation}
where the degree-zero block is one trainable vector per element type, shared by all atoms of that type.
The subscript \(\ell\) labels the angular degree and the parenthesized superscript labels the aggregation layer.
Because the initial features are invariant scalars, the messages depend on a neighbor only through its species and its relative position.

\paragraph{Analytic radial basis.}
The message carried by an edge \(j\to i\) factorizes into a learned radial--chemical amplitude, a block of fixed angular harmonics and the cutoff envelope.
The amplitude is built first, starting from a fixed expansion of the distance.
A scalar distance is a poor input to a narrow network, so it is expanded in \(N_{\mathrm{rbf}}\) fixed analytic functions collected in \(\bff(\rho)\in\R^{N_{\mathrm{rbf}}}\).
DPA4C reuses the DPA4 basis~\cite{li2026dpa4}, which provides two families.
The Bessel family, following the spherical Bessel expansion introduced for directional message-passing potentials~\cite{gasteiger2019directional}, is
\begin{equation}
    f_n(\rho)=\omega_n\sinc\!\left(\frac{\omega_n\rho}{\pi}\right)=\frac{\sin(\omega_n\rho)}{\rho},
    \qquad n=1,\dots,N_{\mathrm{rbf}},
    \label{eq:bessel-basis}
\end{equation}
with \(\sinc(z)=\sin(\pi z)/(\pi z)\) and \(\sinc(0)=1\), so that \(f_n\) tends to \(\omega_n\) as \(\rho\to0\) and remains differentiable there.
The Gaussian family is \(f_n(\rho)=\exp[-(\rho-c_n)^{2}/(2\sigma_{\mathrm{g}}^{2})]\) with fixed width \(\sigma_{\mathrm{g}}=\rcut/(N_{\mathrm{rbf}}-1)\).
The frequencies \(\omega_n\) and the centers \(c_n\) are trainable.
In contrast to the DPA4 default, the basis here is evaluated without a cutoff factor, because DPA4C applies one explicit envelope after the radial and chemical information have been combined.

\paragraph{Shared radial map.}
The basis is passed through a bias-free feed-forward network with a single gated hidden layer.
The gating unit is the SiLU-gated linear unit (SwiGLU), which splits a linear map into a gate branch and a value branch, applies the sigmoid-weighted linear unit (SiLU) \(\SiLU(z)=z/(1+e^{-z})\) to the gate and multiplies the branches elementwise~\cite{hendrycks2016gaussian,elfwing2018sigmoid,shazeer2020glu}:
\begin{equation}
    \SwiGLU(\bm x;\bW)=\SiLU(\bm x_{\mathrm{g}})\odot\bm x_{\mathrm{v}}\in\R^{H},
    \qquad
    [\bm x_{\mathrm{g}},\bm x_{\mathrm{v}}]=\bm x\bW\in\R^{2H},
    \label{eq:swiglu}
\end{equation}
where \(\bm x\) is the input, \(\bW\) a weight matrix with \(2H\) columns and \(\odot\) the elementwise product.
The radial branch is
\begin{equation}
    \bh_{ij}=\SwiGLU\bigl(\bff(\rho_{ij});\bW_{\mathrm{in}}\bigr)\in\R^{H},
    \qquad
    \bg(\rho_{ij})=\bh_{ij}\bW_{\mathrm{out}}\in\R^{C_0},
    \label{eq:radial-map}
\end{equation}
with \(\bW_{\mathrm{in}}\in\R^{N_{\mathrm{rbf}}\times2H}\), \(\bW_{\mathrm{out}}\in\R^{H\times C_0}\) and post-gate hidden width \(H=8\lceil C_0/3\rceil\), the standard width convention for gated feed-forward blocks of model width \(C_0\)~\cite{shazeer2020glu}.
When the number of shared radial modes \(R\), the third structural control of Section~\ref{sec:methods-setting}, is positive, a second linear head branches off the same hidden state and produces \(R\) additional distance profiles,
\begin{equation}
    \bq(\rho_{ij})=\bh_{ij}\bW_{\mathrm{mode}}\in\R^{R},
    \qquad
    \bW_{\mathrm{mode}}\in\R^{H\times R}.
    \label{eq:radial-modes}
\end{equation}
Both \(\bg\) and \(\bq\) are functions of the scalar \(\rho\) alone and are shared by every element pair and every angular degree.
The role of the modes is fixed by the pair modulation defined next.

\paragraph{Ordered type-pair modulation.}
Different ordered element pairs require different radial responses, and DPA4C obtains them by feature-wise linear modulation, in which a conditioning input produces a per-channel scale and shift rather than a full transformation~\cite{perez2018film}.
The initial features of the two endpoints, Eq.~\eqref{eq:init-features}, are concatenated into \(\bz_{ij}=[\bX^{(0)}_{i,0},\bX^{(0)}_{j,0}]\in\R^{2C_0}\) and passed through a second bias-free gated network of the form of Eq.~\eqref{eq:swiglu}, with hidden width \(8\lceil2C_0/3\rceil\) and a fixed output scale of \(0.1\):
\begin{equation}
    \begin{aligned}
        [\bm s_{ij},\bm t_{ij},\bm w_{ij}]
                    & =0.1\,\SwiGLU\bigl(\bz_{ij};\bW_{\mathrm{p,in}}\bigr)\bW_{\mathrm{p,out}}\in\R^{C_0(2+R)}, \\
        \bgamma_{ij}=\bm1+\tanh(\bm s_{ij}),
        \qquad
        \bbeta_{ij} & =\bX^{(0)}_{i,0}+\bX^{(0)}_{j,0}+\tanh(\bm t_{ij}),
        \qquad
        \bU_{ij}=\operatorname{reshape}\bigl[\tanh(\bm w_{ij})\bigr],
    \end{aligned}
    \label{eq:pairfilm}
\end{equation}
where \(\bm1\) is the all-ones vector and the block \(\bm w_{ij}\) is present only when \(R>0\).
The bounded activations place \(\bgamma_{ij}\) in \((0,2)^{C_0}\) and the learned parts of \(\bbeta_{ij}\in\R^{C_0}\) and \(\bU_{ij}\in\R^{C_0\times R}\) in \((-1,1)\), which keeps the cached tables well conditioned in single precision.
These coefficients depend on the two atoms only through their species, the ordered pair \((a_i,a_j)\).
The modulation is ordered, so \((a,b)\) and \((b,a)\) are treated as distinct inputs, and it takes at most \((T+1)^{2}\) distinct values, where the extra row and column belong to a zero padding type.
At deployment the network is therefore evaluated once per ordered type pair and cached as a finite table, never per edge.

\paragraph{Edge amplitude.}
The scale, the shift and the mode mixing combine into
\begin{equation}
    \psi_{ij,c}
    =\gamma_{ij,c}\,g_c(\rho_{ij})
    +\beta_{ij,c}
    +\sum_{\mu=1}^{R}
    U_{ij,c\mu}\,q_\mu(\rho_{ij}), \qquad c=1,\dots,C_0.
    \label{eq:amplitude}
\end{equation}
With \(R=0\) every ordered pair applies a per-channel affine transformation to one shared radial function, so pairs differ through both the scale and the offset of their radial response.
With \(R>0\), every channel additionally receives a linear combination of the \(R\) shared radial profiles, with mixing coefficients set by the ordered species pair.
Different pairs can therefore shape, not merely scale, their radial response.
Equation~\eqref{eq:amplitude} is the finite sum of radial--chemical products required by the compressibility constraint of Section~\ref{sec:methods-setting}.

Distances alone cannot resolve how neighbors are arranged around a center.
The angular factor supplies this information, and the aggregation combines the two into weighted sums over the neighbors that replace a variable-size neighborhood by a fixed-size equivariant tensor.

\paragraph{Real Cartesian harmonics.}
DPA4C uses real solid harmonics written as low-degree polynomials of the Cartesian components of \(\bu\), which avoids complex arithmetic and admits direct evaluation.
The block of degree \(\ell\) is \(\bB_\ell(\bu)\in\R^{2\ell+1}\), with components \(B_{\ell m}\) indexed by \(m=1,\dots,2\ell+1\).
Through degree two, with \(\bu=(u_x,u_y,u_z)\),
\begin{equation}
    \begin{aligned}
        \bB_0 & =(1),\qquad\bB_1=(u_x,u_y,u_z),                                                                                                                              \\
        \bB_2 & =\Bigl(\sqrt{3}u_xu_y,\;\sqrt{3}u_yu_z,\;\tfrac{1}{2}\bigl(3u_z^{2}-\lVert\bu\rVert^{2}\bigr),\;\sqrt{3}u_xu_z,\;\tfrac{\sqrt{3}}{2}(u_x^{2}-u_y^{2})\Bigr).
    \end{aligned}
    \label{eq:harmonics}
\end{equation}
Degree one is listed in Cartesian order and degree two in increasing \(m\) order.
Degrees three and four are the real solid harmonics of the same normalization in the same order, and are given in Supplementary Note~\ref{si:harmonics}.
The normalization is fixed by the addition theorem: for any nonzero \(\bu\) and \(\bv\),
\begin{equation}
    \bB_\ell(\bu)\cdot\bB_\ell(\bv)
    =\bigl(\lVert\bu\rVert\,\lVert\bv\rVert\bigr)^{\ell}
    P_\ell\!
    \left(\frac{\bu\cdot\bv}{\lVert\bu\rVert\,\lVert\bv\rVert}\right),
    \label{eq:addition-theorem}
\end{equation}
where \(P_\ell\) is the Legendre polynomial of degree \(\ell\).
The blocks are equivariant.
An orthogonal transformation of \(\bu\) mixes the \(2\ell+1\) components of a block among themselves, never across degrees, and each block carries an irreducible representation of \(\Orth(3)\)~\cite{varshalovich1988quantum}.
Under inversion they acquire only a parity sign,
\begin{equation}
    \bB_\ell(-\bu)=(-1)^{\ell}\bB_\ell(\bu).
    \label{eq:parity}
\end{equation}

\paragraph{Matrix form of degree two.}
Degree-two quantities are also used in matrix form.
Let \(\STF:\R^{5}\to\R^{3\times3}\) be the linear map fixed by
\begin{equation}
    \STF\bigl[\bB_2(\bu)\bigr]=\sqrt{\tfrac32}\,\Bigl(\bu\bu\tp-\tfrac13\lVert\bu\rVert^{2}\bI\Bigr),
    \label{eq:stf}
\end{equation}
which determines \(\STF\) uniquely because the vectors \(\bB_2(\bu)\) span \(\R^{5}\).
Its image consists of symmetric trace-free (STF) matrices, and it is an isometry between the Euclidean inner product on \(\R^{5}\) and the Frobenius inner product \(\bQ:\bQ'=\tr(\bQ\tp\bQ')\).

\paragraph{Degree-wise widths.}
Each degree retains its own number of channels, read from the leading entries of the shared radial map, so that the map to be tabulated keeps width \(C_0\) irrespective of \(L\).
A non-scalar degree costs \(2\ell+1\) accumulators per channel and produces a Gram block quadratic in its width, so the widths follow \(C_0\) sublinearly:
\begin{equation}
    C_1=\max\bigl(4,\,2^{\lceil\frac12\log_2 C_0\rceil}\bigr), \qquad C_2=\max\bigl(4,\,C_1/2\bigr), \qquad C_\ell=1\ \ (\ell\ge3).
    \label{eq:degree-channels}
\end{equation}
Degrees one and two keep the widest channel blocks, while degrees three and four keep one channel each, so that raising \(L\) adds angular resolution without letting the highest degrees dominate the node-feature state or the readout.
The floor of four binds at degree two for every \(C_0\) up to 64, so over that range widening the model enlarges the scalar block and the degree-one block alone.
Angular capacity is bought by raising \(L\), not by raising \(C_0\).
These sublinear widths are the compactness constraint of Section~\ref{sec:methods-setting} at work.
The flat width of that state is
\begin{equation}
    S=\sum_{\ell=0}^{L}(2\ell+1)C_\ell.
    \label{eq:moment-width}
\end{equation}
For \(C_0=32\) and \(L=2\), this gives \(C_1=8\), \(C_2=4\) and \(S=76\).
Supplementary Table~\ref{si:tab-widths} lists the derived widths for every supported combination of \(C_0\) and \(L\).

\paragraph{Envelope weights and their normalizers.}
The envelope enters the aggregation through a degree-dependent edge weight,
\begin{equation}
    \chi_{ij,0}=\chi_{ij},
    \qquad
    \chi_{ij,\ell}=\chi_{ij}^{2}\quad(1\le\ell\le L),
    \label{eq:degree-weight}
\end{equation}
so that the scalar features carry one envelope factor and every non-scalar feature carries two.
The second factor concentrates the angular information on nearer neighbors.
Each degree is normalized by the Euclidean norm of its own weight profile over the neighborhood, raised off zero by a fixed floor,
\begin{equation}
    M_{i,\ell}=\Bigl(\tfrac14+\sum_{j\in\Nb_i}\chi_{ij,\ell}^{2}\Bigr)^{1/2}\ \ge\ \tfrac12 .
    \label{eq:masses}
\end{equation}
Only two distinct values occur, \(M_{i,0}\) for the scalar block and \(M_{i,1}=\dots=M_{i,L}\) for the non-scalar blocks.
These are smooth measures of neighborhood weight rather than integer coordination numbers, because a neighbor near the cutoff contributes only fractionally, and the floor bounds the rescaling that an almost empty environment can produce.
For the non-scalar degrees, the direction-dependent harmonic factors of different neighbors partially cancel, so the aggregated features grow only as the square root of the neighbor count, at the same rate as the normalizers.
Dividing by the normalizers therefore removes their leading dependence on coordination.
The scalar contributions carry no direction factor and add with a single sign, so for \(n\) equivalent neighbors the scalar feature approaches \(\sqrt{n}\,\psi_{c}\), where \(\psi_c\) is the amplitude shared by the equivalent edges, and the scalar block retains a coordination signal of its own.

\paragraph{One aggregation layer.}
The complete message of an edge packs every quantity that must be accumulated into one flat vector,
\begin{equation}
    \bm m_{ij}=
    \Bigl[\;\chi_{ij,0}^{2},\;\;\chi_{ij,1}^{2},\;\;
        \bigl\{\chi_{ij,\ell}\,\psi_{ij,c}\,B_{\ell m}(\bu_{ij})\bigr\}_{0\le\ell\le L,\;m\le2\ell+1,\;c\le C_\ell}\;\Bigr]
    \in\R^{S+2},
    \label{eq:payload}
\end{equation}
with \(B_{01}\equiv1\), and a single neighbor reduction, that is one segment sum of the messages over the destination index of the graph, produces both normalizing scales and every unnormalized feature.
Dividing by the scales gives the node features after aggregation,
\begin{equation}
    X^{(1)}_{i,\ell,m,c} =\frac{1}{M_{i,\ell}}\sum_{j\in\Nb_i}\chi_{ij,\ell}\,\psi_{ij,c}\,B_{\ell m}(\bu_{ij}), \qquad 0\le\ell\le L, \label{eq:moments} \end{equation} so that \(\bX^{(1)}_{i,\ell}\) has shape \((2\ell+1,C_\ell)\) and the degree-zero harmonic axis has length one.
Hereafter we drop the layer superscript and write \(\bX_{i,\ell}\equiv\bX^{(1)}_{i,\ell}\).
The features are kept in the flat layout of Eq.~\eqref{eq:payload}, ordered by degree, then by harmonic component, then by channel.
Because neighbors are combined only by these sums, the features do not depend on the order in which neighbors are stored.
Under an orthogonal transformation of the configuration, each degree transforms with the corresponding representation of \(\Orth(3)\).
Supplementary Note~\ref{si:symmetry-proof} records the proof.

Equation~\eqref{eq:moments} is the single message-passing layer of DPA4C~\cite{gilmer2017neural}.
Its messages are cheap.
Because the initial node features are invariant scalars fixed by the species, a message depends on its neighbor only through the species and the relative position.
Each term of the sum involves one neighbor at a time.
Correlations among neighbors arise only later, from the polynomial contractions of the readout.
Every operation after Eq.~\eqref{eq:moments} is local to a single center.
A second layer would have to store and communicate the aggregated equivariant features themselves, and DPA4C stops at one, so no learned feature ever travels between atoms.

\subsection{Nonlinear readout and atomic energy}
\label{sec:methods-readout}

The readout converts the equivariant node features into the invariant feature vector \(\bD_i\) and maps it to the atomic energy.
Invariants are formed by contracting the equivariant node features over their harmonic indices with fixed coefficients.
The learned maps in this stage act only on channel indices and are linear, so the MLP at the end remains the only trainable nonlinearity.
The requirement guiding the construction is that the retained invariants should determine the node features up to a global rotation or reflection, the ambiguity that no invariant can resolve.
Classical invariant theory makes this a finite task.
The ring of polynomial invariants of any finite list of feature blocks admits a finite generating set.
For vector and symmetric trace-free matrix channels, the invariants of degree at most four are known explicitly~\cite{spencer1958theory}.
For computational efficiency, DPA4C retains the quadratic generators in full and restricts the cubic and quartic generators through learned degree-wise low-rank projections, selecting within a closed finite family.
A contraction of \(\nu\) feature blocks is at the same time a sum over \(\nu\)-tuples of neighbors and describes correlations of body order up to \(\nu+1\), counting the center.
This places the retained set in the vocabulary of cluster-expansion and bispectrum models~\cite{drautz2019atomic,bartok2013representing,thompson2015spectral}.

\paragraph{Channel alignment and exact Gram invariants.}
Degrees one and two first pass through a full-width residual channel map,
\begin{equation}
    \widetilde\bX_{i,\ell}=\bX_{i,\ell}\bigl(\bI+\bW_\ell\bigr),
    \qquad \bW_\ell\in\R^{C_\ell\times C_\ell},
    \qquad \ell\in\{1,2\},
    \label{eq:alignment}
\end{equation}
which acts on the channel axis only and therefore preserves equivariance.
The identity term provides an explicit residual channel path in addition to the learned mixing.
Because the degree blocks are irreducible and pairwise inequivalent, channel mixing of this form, acting identically on every component \(m\), is the most general equivariant linear map on the features.
For \(\ell\ge3\), \(\widetilde\bX_{i,\ell}=\bX_{i,\ell}\).

Within a fixed angular degree, the channel Gram matrix contains the complete quadratic \(\Orth(3)\)-invariant information.
Treating each channel as a vector of length \(2\ell+1\) over harmonic components,
\begin{equation}
    G_{i,\ell,cc'}=\sum_{m=1}^{2\ell+1}\widetilde X_{i,\ell,m,c}\widetilde X_{i,\ell,m,c'}, \qquad \bG_{i,\ell}\in\R^{C_\ell\times C_\ell}.
    \label{eq:gram}
\end{equation}
An orthogonal transformation of the harmonic components leaves every inner product unchanged, so \(\bG_{i,\ell}\) is \(\Orth(3)\)-invariant.
Its diagonal is the channel power spectrum and its off-diagonal entries add the cross-channel terms.
These entries exhaust the quadratic generators.
Within a degree every pairwise scalar product is a Gram entry, and across degrees no quadratic invariant exists because the blocks are inequivalent irreducibles.
Only the upper triangle enters the feature vector, with every strict off-diagonal entry multiplied by \(\sqrt2\) so that the half-vectorization \(\vech(\cdot)\) preserves the Frobenius norm and no direction of the Gram matrix is reweighted.

\paragraph{Degree-wise low-rank projections and Cartesian bispectrum.}
Quadratic invariants meet the completeness requirement at degree one but not beyond.
The Gram matrix \(\bG_{i,1}\) fixes the channel vectors up to a common orthogonal transformation of the three Cartesian components, which is exactly the physical freedom.
For degree two, the Gram matrix treats the five packed components as an abstract five-dimensional space and is unchanged under the ten-parameter group \(\Orth(5)\), of which the physical rotations occupy only a three-parameter subgroup.
Moreover, since no quadratic invariant couples one degree to another, the relative orientation between degree blocks is lost entirely.
Third-order contractions, the cubic generators of the invariant ring, recover part of this information.
In the vocabulary of the neighbor-density expansion they are the bispectrum~\cite{bartok2013representing,thompson2015spectral}.

Only degree triples that couple to a scalar and are even under inversion contribute, which selects
\begin{equation}
    \mathcal{T}_L=\bigl\{(\ell_1,\ell_2,\ell_3):
    1\le\ell_1\le\ell_2\le\ell_3\le L,\;
    \ell_3\le\ell_1+\ell_2,\;
    \ell_1+\ell_2+\ell_3\ \text{even}\bigr\},
    \label{eq:triples}
\end{equation}
enumerated in lexicographic order.
The inequality \(\ell_3\le\ell_1+\ell_2\) is the angular triangle rule, and the parity condition follows from Eq.~\eqref{eq:parity}.
An odd total degree gives a pseudoscalar, which changes sign under reflection and cannot enter an \(\Orth(3)\)-invariant energy.
Triples containing degree zero are omitted because they reduce to products of quantities the feature vector already carries.
For \(L=4\) the admissible triples are \((1,1,2)\), \((1,2,3)\), \((1,3,4)\), \((2,2,2)\), \((2,2,4)\), \((2,3,3)\), \((2,4,4)\), \((3,3,4)\) and \((4,4,4)\).
The sets for \(L=2\) and \(L=3\) are the subsets with \(\ell_3\le L\).

The coupling tensor of a triple is built from Gaunt coefficients, the integrals of a product of three harmonics over the unit sphere \(\Sph\), in the real Cartesian convention of Eq.~\eqref{eq:harmonics}~\cite{homeier1996some}:
\begin{equation}
    \mathcal{I}^{(\ell_1\ell_2\ell_3)}_{m_1m_2m_3}
    =\int_{\Sph}
    B_{\ell_1m_1}(\bu)B_{\ell_2m_2}(\bu)B_{\ell_3m_3}(\bu)\,\mathrm{d}\Omega(\bu), \qquad \mathcal{C}^{(\ell_1\ell_2\ell_3)} =\pm\frac{\mathcal{I}^{(\ell_1\ell_2\ell_3)}}{\lVert\mathcal{I}^{(\ell_1\ell_2\ell_3)}\rVert_{\mathrm F}}, \label{eq:gaunt} \end{equation} where \(\mathrm{d}\Omega\) is the surface measure, the tensor is scaled to unit Frobenius norm, and its overall sign is fixed by a convention that the first layer of the MLP absorbs.
The integrals are evaluated once, when the model is constructed, and are exact because the integrand is a polynomial of degree \(\ell_1+\ell_2+\ell_3\le12\) on the sphere.

Instantiating the cubic generators over all channels would scale cubically in the channel width, against the quadratic cost of the Gram entries, so the third-order terms instead use learned degree-wise low-rank projections.
A probe is a linear combination of channels,
\begin{equation}
    Z_{i,\ell,m,\kappa}=\sum_{c=1}^{C_\ell}\widetilde X_{i,\ell,m,c}A_{\ell,c\kappa},
    \qquad
    \bA_\ell\in\R^{C_\ell\times K_\ell},
    \qquad \kappa=1,\dots,K_\ell,
    \label{eq:probes}
\end{equation}
with ranks fixed by the architecture rather than left free,
\begin{equation}
    K_1=C_2,\qquad K_2=2,\qquad K_\ell=1\ \ (\ell\ge3), \label{eq:probe-ranks} \end{equation}
When \(K_\ell=C_\ell\), the implementation uses the identity.
This occurs for the degree-one block at the two narrowest scalar widths and for every degree \(\ell\ge3\).
The remaining degree-one profiles and every degree-two block use genuine low-rank projections.
The bispectrum features are
\begin{equation}
    J^{(\ell_1\ell_2\ell_3)}_{i,\kappa_1\kappa_2\kappa_3} =\sum_{m_1m_2m_3} \mathcal{C}^{(\ell_1\ell_2\ell_3)}_{m_1m_2m_3} Z^{(\ell_1)}_{i,m_1,\kappa_1} Z^{(\ell_2)}_{i,m_2,\kappa_2} Z^{(\ell_3)}_{i,m_3,\kappa_3}.
    \label{eq:bispectrum}
\end{equation}
Because the coupling tensor is symmetric under permutations of axes with equal degrees, permuting the corresponding probe indices reproduces the same value.
Of each such family of equal entries only the non-decreasing index representative is retained, rescaled by the square root of the family size so that the norm of the full symmetric tensor is preserved.
One triple therefore contributes
\begin{equation}
    D_{\mathrm{bis}}^{(\ell_1\ell_2\ell_3)}=
    \begin{cases}
        K_{\ell_1}K_{\ell_2}K_{\ell_3}, & \text{all degrees distinct}, \\[2pt] \tfrac12 K(K+1)\,K', & \text{exactly two degrees equal}, \\[2pt] \tfrac16 K(K+1)(K+2), & \text{all three degrees equal},\end{cases} \label{eq:bispectrum-count} \end{equation} features, where \(K\) is the rank of the repeated degree and \(K'\) that of the remaining one, and \(D_{\mathrm{bis}}=\sum_{(\ell_1,\ell_2,\ell_3)\in\mathcal{T}_L}D_{\mathrm{bis}}^{(\ell_1\ell_2\ell_3)}\).

Two triples admit closed forms.
For \(L=2\) they are the only admissible triples.
Let \(\bv_\kappa\in\R^{3}\) be the \(\kappa\)-th degree-one probe of atom \(i\), with components \(v_{\kappa,m}=Z_{i,1,m,\kappa}\), and let \(\bQ_\eta=\STF\bigl(Z_{i,2,\cdot,\eta}\bigr)\) be the \(\eta\)-th degree-two probe in matrix form, \(\eta=1,\dots,K_2\).
Then
\begin{equation}
    J^{(112)}_{i,\kappa_1\kappa_2\eta}=-\frac{1}{\sqrt5}\,\bv_{\kappa_1}\tp\bQ_\eta\bv_{\kappa_2}, \qquad J^{(222)}_{i,\eta_1\eta_2\eta_3}=-\sqrt{\frac{12}{35}}\,\tr\bigl(\bQ_{\eta_1}\bQ_{\eta_2}\bQ_{\eta_3}\bigr).
    \label{eq:closed-forms}
\end{equation}
Their invariance is immediate.
Under \(\bu\mapsto\bR\bu\) the degree-one probes transform as \(\bv\mapsto\bR\bv\) and the degree-two probes as \(\bQ\mapsto\bR\bQ\bR\tp\), which leaves both a quadratic form and the trace of a matrix product unchanged.
The framework tensor implementation evaluates \((1,1,2)\) with the first closed form, reusing \(\bQ\bv\) for Eq.~\eqref{eq:quartic}, and evaluates \((2,2,2)\) through the generic coupling contraction of Eq.~\eqref{eq:bispectrum}.
The compressed implementation specializes both identities in Eq.~\eqref{eq:closed-forms} and reserves the sparse coupling path for the remaining triples.
Equation~\eqref{eq:harmonics} fixes the convention at degrees zero through two.
At degrees three and four the choice of orthonormal basis within a degree is free, because those degrees enter only through squared norms and through contractions with a coupling tensor built in the same basis.
A simultaneous orthogonal change of basis in the node features and coupling tensor leaves these contractions unchanged.
The independently fixed overall sign of each normalized coupling sets only the sign convention of the corresponding bispectrum block.

\paragraph{Projected quartic invariant.}
The closed form for the \((1,1,2)\) triple builds the intermediate \(\bQ_\eta\bv_\kappa\).
Contracting it with itself gives a fourth-order invariant at the cost of one further reduction,
\begin{equation}
    \Pi_{i,\eta\kappa}
    =\bigl\lVert\bQ_\eta\bv_\kappa\bigr\rVert^{2}
    =\bv_\kappa\tp\bQ_\eta^{2}\bv_\kappa,
    \qquad
    \bPi_i\in\R^{K_2\times K_1},
    \label{eq:quartic}
\end{equation}
which is invariant because it is a squared vector length and which adds \(K_1K_2\) features of body order up to five.

This term closes the description of a single pair of probes.
For one vector \(\bv\) and one symmetric trace-free matrix \(\bQ\), the polynomial invariants of \(\Orth(3)\) are generated by
\begin{equation}
    \bv\cdot\bv,
    \qquad
    \tr\bQ^{2},
    \qquad
    \tr\bQ^{3},
    \qquad
    \bv\tp\bQ\bv,
    \qquad
    \bv\tp\bQ^{2}\bv,
    \label{eq:integrity-basis}
\end{equation}
five functions for the five degrees of freedom that survive after the three of the group are removed~\cite{spencer1958theory}.
Their values determine the pair up to a global rotation or reflection, so within each probe pair the completeness requirement is met in full.
DPA4C retains all five.
The two quadratic generators are linear combinations of the Gram entries of Eq.~\eqref{eq:gram}, the two cubic generators are the closed forms of Eq.~\eqref{eq:closed-forms}, and the quartic generator is Eq.~\eqref{eq:quartic}.

\paragraph{What the readout represents.}
The retained set is a fixed polynomial family in the node features, of degree two, three and four, evaluated on a learned linear reparameterization of the channels.
Two deliberate restrictions control how the higher-order members of this family are instantiated across channels.
The cubic terms \(\bv_{\kappa_1}\tp\bQ_\eta\bv_{\kappa_2}\) are retained for every probe pair, whereas of the quartic terms only the diagonal \(\kappa_1=\kappa_2\) is, so the per-pair completeness statement does not extend to cross-pair combinations.
Where they reduce rank, the degree-wise low-rank projections further confine the third- and fourth-order terms to a subspace of each degree that is learned once and shared by all atoms.
Channels outside that subspace enter only through their Gram entries.
Together, these restrictions reduce the number of cubic and quartic channel combinations while preserving the full Gram blocks and the five-invariant integrity basis for each individual probe pair.

\paragraph{Assembly and calibration.}
\label{sec:methods-output}
The retained invariant blocks are concatenated, in order of increasing polynomial degree in the messages, into
\begin{equation}
    \widetilde\bD_i=\Bigl[\;
        \bX^{(0)}_{i,0},\;\;
        \bX_{i,0},\;\;
        M_{i,0},\;\;M_{i,1},\;\;
        \bigl\{\vech\bigl(\bG_{i,\ell}\bigr)\bigr\}_{\ell=1}^{L},\;\;
        \bigl\{\bJ^{(\ell_1\ell_2\ell_3)}_i\bigr\}_{\mathcal{T}_L},\;\;
        \bPi_i
        \;\Bigr],
    \label{eq:descriptor}
\end{equation}
where \(\bJ^{(\ell_1\ell_2\ell_3)}_i\) collects the retained entries of Eq.~\eqref{eq:bispectrum} for one triple, the Gram blocks appear in order of increasing degree and the bispectrum blocks in the lexicographic order of Eq.~\eqref{eq:triples}.
Its width is
\begin{equation}
    D_{\mathrm{out}} =2C_0 +2 +\sum_{\ell=1}^{L}\frac{C_\ell(C_\ell+1)}{2} +D_{\mathrm{bis}} +K_1K_2.
    \label{eq:dim-out}
\end{equation}
The two contributions of size \(C_0\) are the initial feature \(\bX^{(0)}_{i,0}\) of the center, which enters as a skip connection past the aggregation layer, and the aggregated scalar block \(\bX_{i,0}\).
The constant counts the two normalizing scales, and for \(C_0=32\) and \(L=2\), \(D_{\mathrm{out}}=144\).
The number of radial modes does not appear.
Increasing \(R\) changes the per-edge work but neither the node-feature state nor the width of the invariant feature vector.

The two normalizing scales enter the feature vector alongside the invariants because they are the only place where the absolute weight of the neighborhood survives Eq.~\eqref{eq:moments} in a form the readout can read directly, and because they depend on the coordinates and therefore contribute to the forces.
They cost no additional reduction, since they are already part of the message in Eq.~\eqref{eq:payload}.

The blocks of Eq.~\eqref{eq:descriptor} are polynomials of different order in the node features and differ widely in scale, so \(\widetilde\bD_i\) passes through a fixed diagonal calibration that yields the invariant feature vector consumed by the MLP,
\begin{equation}
    \bD_i=\bigl(\widetilde\bD_i-\bm\mu\bigr)\oslash\bm\sigma,
    \qquad
    \bm\mu,\bm\sigma\in\R^{D_{\mathrm{out}}},
    \label{eq:calibration}
\end{equation}
with \(\oslash\) denoting elementwise division.
This is an initialization-time preconditioner rather than a running normalization.
The statistics \(\bm\mu\) and \(\bm\sigma\) are estimated once from a sample of training frames and then held fixed.
Every geometric entry is brought to the root-mean-square scale of the initial features.
For all but two of them the stored mean is zero and the stored scale is the measured root mean square.
The two normalizing scales are standardized instead, with their measured mean subtracted and their centered standard deviation used in place of the root mean square, because by Eq.~\eqref{eq:masses} they are bounded below by \(\tfrac12\) and fluctuate about a mean several times their spread.
The initial feature of the center is passed through unchanged.

\paragraph{Atomic energy, forces and the virial.}
The atomic energy is
\begin{equation}
    E_i=\MLP\bigl(\bD_i\bigr)+E^{\mathrm{ref}}_{a_i},
    \label{eq:atomic-energy}
\end{equation}
where \(E^{\mathrm{ref}}_{a}\) is a per-element reference energy fitted to the data and \(\MLP\) has \(\Lambda\) hidden layers of equal width and a linear scalar head:
\begin{equation}
    \begin{aligned}
        \bh^{(1)}_i                   & =\phi\bigl(\bD_i\bm\Theta_1+\bm b_1\bigr),                           \\
        \bh^{(\tau)}_i                & =\phi\bigl(\bh^{(\tau-1)}_i\bm\Theta_\tau+\bm b_\tau\bigr)+\bh^{(\tau-1)}_i, \\
        \MLP\bigl(\bD_i\bigr) & =\bh^{(\Lambda)}_i\bm w_{\mathrm{head}}+b_{\mathrm{head}},
    \end{aligned}
    \label{eq:fitting}
\end{equation}
with \(\tau=2,\dots,\Lambda\) and \(\phi\) an elementwise activation.
The identity residual is carried by every layer whose input and output widths agree, which excludes the first layer whenever \(D_{\mathrm{out}}\) differs from the hidden width.
A single network is shared by all elements, and the element of the center enters through its initial feature \(\bX^{(0)}_{i,0}\) inside \(\bD_i\).
Downstream of the aggregation the model applies no other trainable nonlinearity, so the nonlinear depth acting on the assembled environment resides in this network.

Forces and the virial follow by differentiation.
Since the energy depends on the positions only through the edge displacements, the per-edge gradients \(\partial E/\partial\br_{ij}\in\R^{3}\) determine both:
\begin{equation}
    \bF_k=-\frac{\partial E}{\partial\br_k}
    =\sum_{j\in\Nb_k}\frac{\partial E}{\partial\br_{kj}}
    -\sum_{i:\,k\in\Nb_i}\frac{\partial E}{\partial\br_{ik}},
    \qquad
    \bXi=-\sum_{i}\sum_{j\in\Nb_i}\frac{\partial E}{\partial\br_{ij}}\otimes\br_{ij},
    \label{eq:forces}
\end{equation}
where the first sum collects the edges for which atom \(k\) is the center and the second those for which it is the neighbor, \(\otimes\) is the outer product, \(\bF_k\) is in eV/\Ang{} and \(\bXi\) in eV.
Supplementary Note~\ref{si:verification} reports the numerical checks of the identities and invariances used above.

\subsection{Compressed inference}
\label{sec:methods-compression}

Deployment uses a tabulated and compiled form of the same model.
The dominant cost of a step is the edge computation.
A message is evaluated once per directed edge but the readout once per atom, and at the 6~\Ang{} cutoff a condensed-phase atom has of order one hundred neighbors, so message evaluations outnumber readout evaluations by that factor.
Compression therefore optimizes only the edge computation.
The readout is evaluated exactly as defined above.
Tabulation replaces a learned function of one variable, evaluated once per edge, by interpolation in a precomputed table.
DP Compress uses it to accelerate the per-edge embedding networks of Deep Potential models~\cite{lu2022dp,zhang2018end}.
DPA4C extends the strategy through the separable form of its messages (Section~\ref{sec:methods-moments}).
The angular factor, a closed-form block of \(\Orth(3)\) irreducibles, carries no learned parameters, while the learned dependence is confined to the one-dimensional distance and the finite type pair and is therefore tabulated and cached.

\paragraph{Separability.}
The edge computation splits into three parts with different domains.
This separation is the compressibility constraint of Section~\ref{sec:methods-setting} made concrete.
The analytic basis and the radial network of Eqs.~\eqref{eq:bessel-basis}--\eqref{eq:radial-modes} depend only on the scalar \(\rho\), so the composed maps \(\bg(\rho)\) and \(\bq(\rho)\) can be tabulated on \([0,\rcut]\).
The ordered type-pair modulation of Eq.~\eqref{eq:pairfilm} depends only on the finite ordered type pair, so the coefficients \(\bgamma\), \(\bbeta\) and \(\bU\) can be evaluated once per pair and cached.
The envelope and the harmonics are evaluated from their closed forms, so tabulation and caching together capture every learned function on the edge.

\paragraph{Radial table.}
Let \(\Delta\) be a uniform spacing and let \(\rho_s=s\Delta\) for \(s=0,\dots,\lceil\rcut/\Delta\rceil\).
All deployment calculations reported here use \(\Delta=0.002~\Ang\).
For each output channel of the concatenation \([\bg,\bq]\), the value, the first derivative and the second derivative are evaluated at the knots in double precision by automatic differentiation of the trained network.
On each interval the table stores the unique quintic polynomial in \(\rho-\rho_s\) that reproduces those three quantities at both ends of the interval.
This two-point Hermite interpolant is standard~\cite{deboor2001splines} and its coefficients are given in Supplementary Note~\ref{si:quintic}.
Matching second derivatives makes the interpolant twice continuously differentiable across knots, so its contribution to the force is continuously differentiable.
Behavior at the cutoff is unaffected, because the envelope of Eq.~\eqref{eq:envelope} is kept in closed form and multiplies the tabulated amplitude, so every edge contribution still vanishes to third order at \(\rcut\).
Contributions at or beyond \(\rcut\) are removed by the envelope, so the table needs no extrapolation region.
Forward evaluation and backward differentiation use the same polynomial, so the reported force is the analytic derivative of the interpolating model rather than a finite-difference or independently interpolated derivative.

The table has \(\lceil\rcut/\Delta\rceil\) rows and \(6(C_0+R)\) entries per row: 3,000 rows for \(\rcut=6~\Ang\), occupying 1.1 MiB in single precision at \(C_0=16\), \(R=0\) and 2.5 MiB at \(C_0=32\), \(R=4\).
The compressed model also stores the ordered scale, shift and mode-mixing caches, the initial-feature table, the readout matrices of Eqs.~\eqref{eq:alignment} and~\eqref{eq:probes}, the coupling tensors and the calibration vectors.
Each of these is a quantity of the finite type table or of the fixed architecture, so none of them grows with the simulated system.

\paragraph{Fused execution and backward recomputation.}
The compiled implementation in DeePMD-kit~\cite{zeng2023deepmd,zeng2025deepmdv3} first converts the graph to a canonical destination-sorted representation.
A source-index array and Cartesian edge-vector array are accompanied by a CSR row pointer over destination atoms.
The compressed energy-gradient operator is a single fused kernel that assigns one warp to a destination neighborhood and scans its contiguous edge interval once in the forward pass.
The kernel is compiled separately for each combination of the structural parameters \(C_0\), \(L\) and \(R\), and each compiled specialization is called a profile.
The quintic radial interpolant, ordered-pair modulation, envelope and harmonics are evaluated in registers.
Profile-specific subwarp and channel tiling accumulate the \(S+2\) message directly into the destination feature state.
The same kernel then normalizes the features, evaluates the polynomial invariants and writes the calibrated feature vector \(\bD_i\) of each destination atom.
Unlike the framework tensor implementation, which materializes each intermediate stage as a full-system array, this path writes only the retained per-node state and the feature vector \(\bD_i\) to global memory.

The operator retains the per-node feature state and the two normalizing scales required by the backward pass.
After the MLP and polynomial-invariant vector--Jacobian products have been formed, its edge backward pass recomputes the radial interpolation, pair modulation, envelope and harmonics and returns the Cartesian gradient \(\partial E/\partial\br_{ij}\) for each canonical edge.
A separate force--virial operator then combines destination- and source-sorted CSR views to evaluate Eq.~\eqref{eq:forces} without floating-point atomics.
The complete force path therefore uses two compiled operators rather than a single monolithic energy--force--virial kernel, but it does not construct a framework automatic-differentiation tape.

\paragraph{Node tiling.}
The invariant feature vector and the MLP activations are width dependent but node local.
Confining their lifetime to fixed-size tiles realizes the compactness constraint of Section~\ref{sec:methods-setting}.
They are evaluated in contiguous tiles of at most \(B=131{,}072\) destination atoms by default.
Within one tile, the implementation completes the feature-vector forward evaluation, the MLP forward pass and scalar head, the MLP vector--Jacobian product, and the polynomial-invariant vector--Jacobian product before the workspace is reused for the next tile.
Destination sorting makes the edges of a node tile one contiguous graph span, so the tile requires no gathered edge-index list.
The MLP vector--Jacobian product overwrites the feature-vector workspace after the head has consumed it, and the polynomial-invariant vector--Jacobian product in turn reuses the feature storage before edge recomputation begins.
With \(F_1,\ldots,F_\Lambda\) the MLP hidden widths, \(F_{\max}=\max_\tau F_\tau\) and \(\widetilde B=\min(B,N)\), the width-dependent temporary storage is proportional to
\begin{equation}
    \widetilde B\left[D_{\mathrm{out}}+(S+2)+\sum_{\tau=1}^{\Lambda}F_\tau+qF_{\max}\right],
    \qquad
    q=\begin{cases}
        1, & \Lambda=1, \\
        2, & \Lambda>1,
    \end{cases}
    \label{eq:tiled-workspace}
\end{equation}
where the final term is the one- or two-slot MLP scratch array.
This replaces temporary storage proportional to \(N[D_{\mathrm{out}}+(S+2)+\sum_\tau F_\tau+qF_{\max}]\) in an untiled execution.
The canonical graph, Cartesian edge gradients, forces and virials remain full-system arrays with \(O(N_{\mathrm{edge}}+N)\) storage.

\paragraph{Fused MLP.}
The atomic-energy MLP is compiled together with the tiled feature path.
Its dense layers use cuBLAS matrix multiplications, while the bias, activation and identity-residual operations are fused into the elementwise step that follows each multiplication.
Forward activations alternate between two scratch slots instead of allocating one full node array per layer, and the layer pre-activations required for the vector--Jacobian product are packed into one tile-local buffer.
After the scalar head has been evaluated, the feature-vector array becomes the input-gradient array, as described above.
The deployed hidden widths are multiples of four, allowing these elementwise steps to use aligned four-float vector loads.
Atomic energies are accumulated into a double-precision output even though the feature and MLP weights are single precision.

\paragraph{Fidelity of the compressed model.}
The compiled profiles cover \(C_0\in\{8,16,32,64,128\}\), \(L\in\{2,3,4\}\) and \(R\in\{0,2,4,8\}\) in single precision.
The kernel evaluates the harmonics with \(\lVert\bu\rVert^{2}\) set to one.
For the regularized direction of Eq.~\eqref{eq:edge-geometry}, the difference is proportional to \(\varepsilon^{2}/\rho^{2}\) and lies below single-precision resolution at physical separations.
Across the tested profiles, the tabulated model reproduces the continuous implementation to relative deviations of order \(10^{-7}\).
Supplementary Note~\ref{si:verification} gives the numerical checks.

\subsection{Datasets, training and evaluation protocols}
\label{sec:methods-experiments}

\paragraph{OMat24 data.}
OMat24 is an inorganic-materials dataset~\cite{barroso2024open}.
We use its published training and validation sets without subsampling.
The validation set is used only for evaluation in the results reported here.
We do not treat it as an independent test set.
The DPA4 and DPA4C evaluations retain the energies, Cartesian forces and periodic-cell virials supplied with the dataset.

\paragraph{MatPES data.}
MatPES is an inorganic-materials dataset with single-point PBE and r\(^2\)SCAN labels~\cite{kaplan2025matpes}.
We use its r\(^2\)SCAN track in the R2SCAN-2025.2 release, with the published split of 347,889 training and 19,328 test structures.
The DPA4 reference values use the same test split and metric definitions.

\paragraph{OMol25 data.}
OMol25 is a molecular dataset of more than 100 million DFT single-point calculations at the \(\omega\)B97M-V/def2-TZVPD level computed with ORCA~\cite{levine2025open}.
We use its OMol-0 release.
It combines biomolecules, metal complexes, electrolytes and recomputed community datasets, with total charge and spin multiplicity supplied as explicit model inputs.
We train on the 101.7-million-structure training split and evaluate on the out-of-distribution composition validation split.
The DPA4 reference values use the same split and metric definitions~\cite{li2026dpa4}.

\paragraph{DPA4C training.}
All five DPA4C variants use a cutoff of \(6~\Ang\), single-precision parameters and the same 118-element type map.
Each run uses NVIDIA H20 GPUs, with the number of GPUs for each model listed in the Supplementary Information.
Training minimizes a weighted MAE over the available energy, force and virial labels.
The OMol25 virial weight is zero.
The optimizer is HybridMuon, which applies the orthogonalized momentum update of Muon~\cite{jordan2024muon} to matrix parameters and Adam to the remaining parameters.
Shared model and optimizer settings are listed in Supplementary Table~\ref{si:tab-model-variants}.
Dataset- and model-specific batch sizes, learning-rate schedules and training lengths are given in Supplementary Tables~\ref{si:tab-omat24-training}, \ref{si:tab-matpes-training} and~\ref{si:tab-omol25-training}.
Training cost is reported in H20 GPU-hours for each complete optimization schedule.

\paragraph{Accuracy metrics.}
For an OMat24 or MatPES evaluation split of \(M\) structures, the energy metric is
\begin{equation}
    \mathrm{MAE}_{E/N}=\frac{1}{M}\sum_{s=1}^{M}\left|\frac{E_s-E_s^{\mathrm{ref}}}{N_s}\right|,
    \label{eq:energy-mae}
\end{equation}
where \(N_s\) is the number of atoms in structure \(s\).
For OMat24 and MatPES, the force MAE is the mean absolute difference over all atoms and Cartesian components in the split.
Stress is computed from the virial as \(\bm\sigma=-\bXi/V\), with \(V\) the cell volume, and its MAE is averaged over all nine Cartesian components and structures.
For these two materials benchmarks, the values are reported in meV/atom, meV/\Ang{} and meV/\Ang\(^{3}\), respectively.
For OMol25, we report the unnormalized total-energy MAE, \(M^{-1}\sum_{s=1}^{M}|E_s-E_s^{\mathrm{ref}}|\), and the force MAE on the out-of-distribution composition validation split using the benchmark aggregation reported in the DPA4 study.
The two OMol25 metrics are reported in kcal/mol and kcal/mol/\Ang{}, respectively.
Stress is not reported.
The DPA4 accuracy values on all three benchmarks are taken from the DPA4 study~\cite{li2026dpa4}.
The MACE-Omat validation values are taken from the independent ASE evaluation of the released MACE-OMAT-0 checkpoints reported in the DPA4 study~\cite{batatia2025crosslearning,li2026dpa4}.
The released NEP89 model~\cite{liang2026nep89} is evaluated independently on the OMat24 validation set with the same aggregation rules.
Its model-specific reference treatment and exact software revisions are reported in Supplementary Note~\ref{si:nep89-omat24-evaluation}.

\paragraph{Whole-step MD benchmark.}
The benchmark systems are periodic diamond-carbon supercells generated from the eight-atom conventional cell with lattice constant \(3.567~\Ang\).
Independent Gaussian coordinate perturbations with standard deviation \(0.03~\Ang\) and seed 0 break exact crystal symmetry.
All measurements use one NVIDIA H20.
DPA4C is evaluated with LAMMPS/Kokkos~\cite{thompson2022lammps,trott2022kokkos}, whereas NEP89 uses GPUMD on the same initial supercells.
Both engines propagate NVT trajectories at 300~K with a time step of 1~fs and a \(1~\Ang\) neighbor skin.
The reported throughput is \(N/t_{\mathrm{step}}\) in atoms/s, where \(t_{\mathrm{step}}\) is the wall time of the complete MD step, including graph construction, model evaluation, force and virial assembly and integration.
We report the saturated throughput, the plateau reached in each system-size scan.

Because each model is measured in its native simulation engine, the comparison retains engine-level overhead and is not an isolated neural-network kernel timing.
For each model, the largest completed system is the largest atom count for which the NVT simulation completes.
The next scanned size is reported as the first failure.
The node-tiling and whole-step operator ablations follow the protocols of Supplementary Note~\ref{si:benchmark-tables}.

\paragraph{Distributed scaling.}
The distributed benchmark uses NVIDIA V100-SXM2-16GB GPUs, with 16 GPUs per node and one MPI process per GPU.
The 32-GPU allocation is the first point spanning two nodes.
The intra-node topology, rank binding and software environment are specified in Supplementary Note~\ref{si:multigpu-scaling}.
All runs use the diamond-carbon construction and the complete-step NVT protocol described above.
DPA4C strong-scaling values use the arithmetic mean across ten independent Slurm allocations at every model--GPU-count point.
Weak-scaling and DPA4-Mini values use the median of three allocations.
Their repeat ranges, together with the DPA4C strong-scaling dispersion, are reported in the Supplementary Information.

Weak scaling uses \(2{,}000{,}376\) atoms per GPU at \(p=1,2,4,\ldots,1024\), so the total atom count on \(p\) GPUs is \(N_p=2{,}000{,}376\,p\) and the 1,024-GPU endpoint contains \(2{,}048{,}385{,}024\) atoms on 64 nodes.
Each weak-scaling run uses 100 warm-up steps and 500 timed steps.
If \(\Theta_p(N)\) denotes throughput on \(p\) GPUs at atom count \(N\), the complete-step time is \(t_p=N_p/\Theta_p(N_p)\) and weak-scaling efficiency is
\begin{equation}
    \eta_{\mathrm{w}}(p)=\frac{t_1}{t_p}.
    \label{eq:weak-scaling}
\end{equation}

Strong scaling fixes one globally cubic \(2{,}000{,}376\)-atom system for every DPA4C variant and measures powers-of-two GPU counts from 1 to 1,024.
Per-variant warm-up and timed-step schedules, the processor grids and the LAMMPS execution contracts are specified in Supplementary Note~\ref{si:multigpu-scaling}.
At the fixed atom count \(N\), the strong speedup and efficiency are
\begin{equation}
    S_{\mathrm{s}}(p) =\frac{\Theta_p(N)}{\Theta_1(N)}, \qquad \eta_{\mathrm{s}}(p) =\frac{S_{\mathrm{s}}(p)}{p}.
    \label{eq:strong-scaling}
\end{equation}
The measured step rate and the 1-fs time step define the MD speed in ns/day.

DPA4-Mini serves as a message-passing deployment reference with its own atom counts, specified in Supplementary Note~\ref{si:multigpu-scaling}.
The unequal workloads are excluded from absolute-throughput comparisons.

\paragraph{Single-GPU crystal scans with empirical potentials.}
The computational-reference benchmark evaluates the five DPA4C variants and NEP89~\cite{liang2026nep89} together with carbon MEAM~\cite{baskes1992meam} and Tersoff~\cite{tersoff1989multicomponent}, or copper EAM~\cite{mishin2001copper} and MEAM, on one Tesla V100-SXM2-16GB GPU.
Diamond-carbon and FCC-copper systems are constructed as near-cubic conventional-cell supercells and propagated in NVT at 300~K with a 1-fs time step.
The LAMMPS paths use a \(1.0~\Ang\) neighbor skin, whereas GPUMD retains its native neighbor handling.
For each material--model pair, three independent scans double the requested atom count from 128 until the first OOM and then perform three bisections.
Every point uses ten warm-up and 100 timed complete steps.
DPA4C and the empirical potentials run with LAMMPS/Kokkos, whereas NEP89 uses native GPUMD.
The empirical potentials retain their native physical cutoffs, and predictive accuracy lies outside this computational comparison.
Exact crystal constructions, execution paths, parameter records, scan aggregation and engine inputs are specified in Supplementary Note~\ref{si:classical-reference}.

%% file: chap/acknowledgments.tex
\section{Acknowledgments}

We gratefully acknowledge the support received for this work.
The distributed-scaling calculations on NVIDIA V100 GPUs were supported by the Open Source Supercomputing Center of S-A-I.
The work of Han Wang is supported by the National Key R\&D Program of China (Grant No.~2022YFA1004300) and the National Natural Science Foundation of China (Grants No.~12525113 and No.~12561160120).
The work of Linfeng Zhang was in part supported by the Advanced Materials-National Science and Technology Major Project, China (No.~2024ZD0606900).
The work of Jianming Xue and Tiancheng Li is supported by the National Natural Science Foundation of China (Grant No.~12135002).

%% file: chap/data_availability.tex
\section{Data and code availability}

The OMat24, MatPES and OMol25 datasets are publicly available from the sources cited in Methods.
The DPA4C training and inference codes are publicly available in the DeePMD-kit repository (\url{https://github.com/deepmodeling/deepmd-kit}) from version 3.2.0.

%% file: chap/si.tex
\beginsupplement

\FloatBarrier
\begin{center}
    {\LARGE\bfseries Supplementary Information for

        \textit{Universal Machine-learning Molecular Dynamics at the Speed of Empirical Potentials}\par}
    \vspace{1em}
\end{center}

\addtocontents{toc}{\protect\setcounter{tocdepth}{2}}
\tableofcontents

\section{Mathematical formulation and validation}
\label{si:mathematical-formulation}

The notation, higher-degree harmonic blocks, derived feature widths, invariance proof and numerical checks below complete the mathematical definition of DPA4C.

\subsection{Notation}
\label{si:notation}

Table~\ref{si:tab-notation} collects the symbols used in the definition of DPA4C.
Italic symbols denote scalars, bold lowercase symbols vectors, bold uppercase symbols matrices and higher-order tensors, and calligraphic capitals fixed index sets and fixed angular tensors.
Group names, operators and function names are set upright.
Italic subscripts are indices and upright subscripts are labels.

\begin{table}[htb]
    \PaperTableStyle
    \caption{Symbols used in the definition of DPA4C.}
    \label{si:tab-notation}
    \begin{threeparttable}
    \begin{tabular*}{\linewidth}{@{\extracolsep{\fill}}L{3.8cm}L{9.4cm}@{}}
        \toprule
        \textbf{Symbol} & \textbf{Meaning} \\
        \midrule
        \(i,\,j,\,k\) & center atom, one of its neighbors and an arbitrary atom \\
        \(a,\,b\) & element types of the center atom \(i\) and of the neighbor atom \(j\) \\
        \(N,\,N_{\mathrm{edge}},\,T\) & numbers of atoms, directed edges and element types \\
        \(E,\,E_i,\,E^{\mathrm{ref}}_a\) & total, atomic and per-element reference energy, in eV \\
        \(\bF_k,\,\bXi\) & force on atom \(k\), shape \((3,)\) in eV/\Ang{}, and virial, shape \((3,3)\) in eV \\
        \(\br_i,\,\br_{ij}=\br_j-\br_i\) & atomic position and edge displacement, shape \((3,)\), in \Ang \\
        \(\rcut,\,\varepsilon\) & cutoff radius and direction regularizer, in \Ang \\
        \(\rho_{ij},\,\bu_{ij}\) & regularized edge length in \Ang{} and regularized direction, shape \((3,)\) \\
        \(\chi_{ij},\,\chi_{ij,\ell}\) & cutoff envelope and degree-\(\ell\) edge weight, dimensionless \\
        \(N_{\mathrm{rbf}},\,\bff(\rho)\) & number of analytic radial basis functions and the basis, shape \((N_{\mathrm{rbf}},)\) \\
        \(\bh_{ij},\,H\) & hidden activations of the radial network, shape \((H,)\), and their width \\
        \(\bg(\rho),\,\bq(\rho)\) & shared learned radial map and mode profiles, shapes \((C_0,)\) and \((R,)\) \\
        \(\bgamma_{ij},\,\bbeta_{ij},\,\bU_{ij}\) & ordered type-pair scale, shift and mode mixing, functions of \((a_i,a_j)\) alone, shapes \((C_0,)\), \((C_0,)\), \((C_0,R)\) \\
        \(\bpsi_{ij}\) & edge amplitude, shape \((C_0,)\) \\
        \(L,\,\ell,\,m\) & maximum angular degree, a degree \(0\le\ell\le L\), and \(m=1,\dots,2\ell+1\) \\
        \(C_\ell,\,K_\ell,\,R\) & channels at degree \(\ell\), probe rank at degree \(\ell\ge1\), number of radial modes \\
        \(c,\,\kappa,\,\eta,\,\mu\) & channel index, degree-one and degree-two probe index, radial-mode index \\
        \(\bB_\ell(\bu)\) & real Cartesian harmonic block of degree \(\ell\), shape \((2\ell+1,)\) \\
        \(\bX^{(0)}_{i,\ell},\,\bX_{i,\ell},\,\widetilde\bX_{i,\ell}\) & initial, aggregated and channel-aligned node features at degree \(\ell\), shape \((2\ell+1,C_\ell)\); the initial degree-zero block is one trainable vector per element type and the initial higher degrees are zero \\
        \(\bZ_{i,\ell},\,\bA_\ell\) & probe-projected features and their projection, shapes \((2\ell+1,K_\ell)\) and \((C_\ell,K_\ell)\) \\
        \(S,\,\bm m_{ij}\) & width of the node-feature state and the complete edge message, shape \((S+2,)\) \\
        \(M_{i,\ell}\) & degree-\(\ell\) normalizing scale, dimensionless \\
        \(\bG_{i,\ell}\) & channel Gram matrix at degree \(\ell\), shape \((C_\ell,C_\ell)\) \\
        \(\mathcal{T}_L,\,\mathcal{C}^{(\ell_1\ell_2\ell_3)}\) & admissible degree triples and the coupling tensor of one triple \\
        \(\bJ^{(\ell_1\ell_2\ell_3)}_i\) & bispectrum tensor over three probe indices, shape \((K_{\ell_1},K_{\ell_2},K_{\ell_3})\); only entries independent under permutations of equal-degree probe indices are retained \\
        \(\bPi_i\) & projected quartic invariant, shape \((K_2,K_1)\) \\
        \(\widetilde\bD_i,\,\bD_i,\,D_{\mathrm{out}}\) & concatenated invariant blocks of atom \(i\) before calibration, the calibrated invariant feature vector, and their width \\
        \(\bm\mu,\,\bm\sigma\) & calibration shift and scale, shape \((D_{\mathrm{out}},)\) \\
        \(\bh^{(\tau)}_i,\,\Lambda\) & hidden activations of the multilayer perceptron (MLP) at layer \(\tau\), and the number of hidden layers \\
        \(\Delta\) & radial table spacing, in \Ang \\
        \bottomrule
    \end{tabular*}
        \begin{tablenotes}[flushleft]
            \PaperTableNotesStyle
            \item[] Shapes are given for non-scalar quantities.
            The scalar channel width \(C_0\), maximum angular degree \(L\) and number of shared radial modes \(R\) are the free structural parameters; all other widths follow from them.
        \end{tablenotes}
    \end{threeparttable}
\end{table}

\FloatBarrier
\subsection{Real solid harmonics of degrees three and four}
\label{si:harmonics}

Degrees zero through two are given in the main text.
Degrees three and four use the same normalization, fixed by the addition theorem, and the same ordering by increasing \(m\).
Writing \(\bu=(u_x,u_y,u_z)\) and \(s=\lVert\bu\rVert^{2}\),
\begin{equation}
    \bB_3(\bu)=
    \begin{pmatrix}
        \sqrt{5/8}\;u_y\,(3u_x^{2}-u_y^{2})       \\
        \sqrt{15}\;u_xu_yu_z                      \\
        \sqrt{3/8}\;u_y\,(5u_z^{2}-s)             \\
        \tfrac12\,u_z\,(5u_z^{2}-3s)              \\
        \sqrt{3/8}\;u_x\,(5u_z^{2}-s)             \\
        \tfrac12\sqrt{15}\;u_z\,(u_x^{2}-u_y^{2}) \\
        \sqrt{5/8}\;u_x\,(u_x^{2}-3u_y^{2})
    \end{pmatrix}\!,
    \qquad
    \bB_4(\bu)=
    \begin{pmatrix}
        \tfrac12\sqrt{35}\;u_xu_y\,(u_x^{2}-u_y^{2})    \\
        \tfrac14\sqrt{70}\;u_yu_z\,(3u_x^{2}-u_y^{2})   \\
        \tfrac12\sqrt{5}\;u_xu_y\,(7u_z^{2}-s)          \\
        \tfrac14\sqrt{10}\;u_yu_z\,(7u_z^{2}-3s)        \\
        \tfrac18\,(35u_z^{4}-30u_z^{2}s+3s^{2})         \\
        \tfrac14\sqrt{10}\;u_xu_z\,(7u_z^{2}-3s)        \\
        \tfrac14\sqrt{5}\,(u_x^{2}-u_y^{2})(7u_z^{2}-s) \\
        \tfrac14\sqrt{70}\;u_xu_z\,(u_x^{2}-3u_y^{2})   \\
        \tfrac18\sqrt{35}\,(u_x^{4}-6u_x^{2}u_y^{2}+u_y^{4})
    \end{pmatrix}\!.
\end{equation}
These blocks enter the feature vector only through squared norms and through contractions with a coupling tensor defined in the same basis.
An orthonormal change of basis therefore leaves the scalar contractions unchanged when the coupling tensor is transformed consistently.

\subsection{Derived structural widths}
\label{si:widths}

The scalar channel width \(C_0\) and the maximum angular degree \(L\) determine the degree channels \(C_\ell\), the probe ranks \(K_\ell\), the feature width \(S\) and the feature-vector width \(D_{\mathrm{out}}\) through the equations of the main text.
Table~\ref{si:tab-widths} evaluates them for every supported combination.
The number of radial modes \(R\) enters none of these widths.

\begin{table}[htb]
    \PaperTableStyle
    \caption{Structural widths for the supported \((C_0,L)\) combinations.}
    \label{si:tab-widths}
    \begin{threeparttable}
        \begin{tabular*}{\linewidth}{@{\extracolsep{\fill}}*{11}{c}@{}}
            \toprule
            &         &         &         &         & \multicolumn{3}{c}{\textbf{Feature width} \(S\)} & \multicolumn{3}{c}{\textbf{Feature-vector width} \(D_{\mathrm{out}}\)}                                \\
            \cmidrule(lr){6-8}\cmidrule(lr){9-11}
            \(C_0\) & \(C_1\) & \(C_2\) & \(K_1\) & \(K_2\) & \(L=2\)                                & \(L=3\)                                                   & \(L=4\) & \(L=2\) & \(L=3\) & \(L=4\) \\
            \midrule
            8       & 4       & 4       & 4       & 2       & 40                                     & 47                                                        & 56      & 70      & 81      & 93      \\
            16      & 4       & 4       & 4       & 2       & 48                                     & 55                                                        & 64      & 86      & 97      & 109     \\
            32      & 8       & 4       & 4       & 2       & 76                                     & 83                                                        & 92      & 144     & 155     & 167     \\
            64      & 8       & 4       & 4       & 2       & 108                                    & 115                                                       & 124     & 208     & 219     & 231     \\
            128     & 16      & 8       & 8       & 2       & 216                                    & 223                                                       & 232     & 522     & 541     & 557     \\
            \bottomrule
        \end{tabular*}
        \begin{tablenotes}[flushleft]
            \PaperTableNotesStyle
            \item[] \(C_\ell\) is the number of channels retained at degree \(\ell\), \(K_\ell\) is the probe rank used in the third- and fourth-order contractions, \(S\) is the accumulated node-feature width and \(D_{\mathrm{out}}\) is the invariant-feature-vector width consumed by the MLP.
        \end{tablenotes}
    \end{threeparttable}
\end{table}

\FloatBarrier
\subsection{Symmetry of the invariant feature vector} \label{si:symmetry-proof}

The symmetry statement below concerns the invariant feature vector defined by Eqs.~\eqref{eq:edge-geometry}--\eqref{eq:descriptor} in exact arithmetic.
It establishes the invariances required of a scalar interatomic potential, but makes no claim that the finite set of invariants separates every pair of distinct atomic environments.

\begin{proposition}[Translation, permutation and orthogonal invariance]
    \label{si:prop-invariance}
    Let \(\bD_i\) be the DPA4C invariant feature vector of atom \(i\), with radial amplitudes depending only on \(\rho_{ij}\) and the ordered type pair \((a_i,a_j)\), and with the third-order contractions restricted to the even-parity set \(\mathcal{T}_L\) of Eq.~\eqref{eq:triples}.
    For every translation \(\bm t\in\R^3\) and every \(\bR\in\Orth(3)\), applying
    \begin{equation}
        \br_k' = \bR\br_k+\bm t
        \label{si:eq-rigid-action}
    \end{equation}
    to all atoms, and to the periodic cell when present, leaves every matched feature vector unchanged: \(\bD_i'=\bD_i\).
    If \(\pi\) is a relabeling among atoms of the same element and \(\br'_{\pi(k)}=\br_k\), then \(\bD'_{\pi(i)}=\bD_i\).
    Thus the collection of feature vectors is equivariant to relabeling, each feature vector is invariant to the ordering of its neighbors, and the total energy of Eq.~\eqref{eq:energy-decomposition} is invariant under all three operations.
\end{proposition}

\begin{proof}
    A translation cancels from every displacement: \(\br'_{ij}=\br'_j-\br'_i=\br_j-\br_i\).
    Hence \(\rho_{ij}\), \(\bu_{ij}\), the cutoff weights, the ordered-pair amplitudes and every subsequent entry of the feature vector are unchanged.

    For a same-element relabeling, the neighbor set transforms as \(\Nb'_{\pi(i)}=\pi(\Nb_i)\), while \(\br'_{\pi(i)\pi(j)}=\br_{ij}\) and \((a'_{\pi(i)},a'_{\pi(j)})=(a_i,a_j)\).
    Each term in Eq.~\eqref{eq:moments} is therefore carried to an identical term, and the sum is merely reindexed.
    The normalizing scales are reindexed sums of the same weights.
    All readout operations are local to one center, so \(\bD'_{\pi(i)}=\bD_i\).

    It remains to consider \(\bR\in\Orth(3)\).
    Orthogonality gives \(\rho'_{ij}=\rho_{ij}\) and \(\bu'_{ij}=\bR\bu_{ij}\), so every radial, cutoff and type-dependent factor is a scalar invariant.
    For each degree \(\ell\), the real solid harmonics carry an orthogonal representation \(\mathsf{R}_\ell(\bR)\):
    \begin{equation}
        \bB_\ell(\bR\bu)
        =\mathsf{R}_\ell(\bR)\bB_\ell(\bu),
        \qquad
        \mathsf{R}_\ell(\bR)\tp
        \mathsf{R}_\ell(\bR)=\bI.
        \label{si:eq-harmonic-action}
    \end{equation}
    The scalar normalizers are unchanged, and Eq.~\eqref{eq:moments} consequently gives
    \begin{equation}
        \bX_{i,\ell}'
        =\mathsf{R}_\ell(\bR)\bX_{i,\ell}.
        \label{si:eq-moment-action}
    \end{equation}
    Channel alignment and the degree-wise low-rank projections multiply on the channel axis and commute with this action.

    Equation~\eqref{si:eq-harmonic-action} immediately leaves each Gram block invariant:
    \begin{equation*}
        \bG_{i,\ell}'
        =(\widetilde\bX_{i,\ell})\tp
        \mathsf{R}_\ell(\bR)\tp\mathsf{R}_\ell(\bR)
        \widetilde\bX_{i,\ell}
        =\bG_{i,\ell}.
    \end{equation*}
    For a proper rotation, invariance of the surface measure in Eq.~\eqref{eq:gaunt} makes the Gaunt tensor an invariant trilinear form, so every bispectrum entry is unchanged.
    Any improper orthogonal transformation can be written as inversion followed by a proper rotation.
    Under inversion, a triple acquires the factor \((-1)^{\ell_1+\ell_2+\ell_3}\), which equals one for every triple in \(\mathcal{T}_L\).
    The third-order blocks are therefore invariant under the full orthogonal group.

    Finally, the STF map intertwines the degree-two action with conjugation, so the probes in Eq.~\eqref{eq:quartic} transform as \(\bv_\kappa' = \bR\bv_\kappa\) and \(\bQ_\eta'=\bR\bQ_\eta\bR\tp\).
    Thus \(\lVert\bQ_\eta'\bv_\kappa'\rVert^2 =\lVert\bR\bQ_\eta\bv_\kappa\rVert^2 =\lVert\bQ_\eta\bv_\kappa\rVert^2\).
    The degree-zero block, the two normalizing scales and the initial feature of the center are invariant scalars, and fixed componentwise calibration preserves their invariance.
    Every block of Eq.~\eqref{eq:descriptor} is therefore unchanged.
\end{proof}

\subsection{Numerical verification of the invariant feature vector}
\label{si:verification}

The identities and invariances asserted in the main text were checked in double precision on randomly generated configurations.
Table~\ref{si:tab-verification} reports the largest deviation observed for each.
The symmetry tests used a periodic cell of 24 atoms of two elements, except for the inversion test, which used an isolated cluster of 12 atoms so that the periodic images do not confound the transformation.
The agreement between the tabulated and the continuous feature vector was measured in single precision at a table spacing of \(0.002~\Ang\), as a maximum deviation relative to the largest feature-vector entry.

\begin{table}[htb]
    \PaperTableStyle
    \caption{Numerical verification of DPA4C identities and invariances.}
    \label{si:tab-verification}
    \begin{threeparttable}
        \begin{tabular*}{\linewidth}{@{\extracolsep{\fill}}lcc@{}}
            \toprule
            \textbf{Property}                           & \textbf{Test}            & \textbf{Deviation}            \\
            \midrule
            Addition theorem, degrees zero to four      & random vector pairs      & \(1\times10^{-12}\)           \\
            Parity of the harmonic blocks               & random directions        & \(0\)                         \\
            Isometry of the symmetric trace-free map    & random packed vectors    & \(2\times10^{-15}\)           \\
            Quadrupolar form of the degree-two block    & random directions        & \(2\times10^{-14}\)           \\
            Closed forms against the Gaunt contraction  & random probe blocks      & \(3\times10^{-13}\)           \\
            Single aggregation against per-degree sums  & 24-atom periodic cell    & \(9\times10^{-16}\)           \\
            Invariance under proper rotation            & 24-atom periodic cell    & \(3\times10^{-15}\)           \\
            Invariance under translation                & 24-atom periodic cell    & \(3\times10^{-15}\)           \\
            Invariance under permutation of atom labels & 24-atom periodic cell    & \(9\times10^{-16}\)           \\
            Invariance under inversion                  & 12-atom isolated cluster & \(0\)                         \\
            Tabulated against continuous feature vector & \((C_0,L,R)=(16,2,0)\)   & \(1.8\times10^{-7}\) relative \\
            & \((C_0,L,R)=(32,2,4)\)   & \(1.3\times10^{-7}\) relative \\
            & \((C_0,L,R)=(32,3,2)\)   & \(1.6\times10^{-7}\) relative \\
            \bottomrule
        \end{tabular*}
        \begin{tablenotes}[flushleft]
            \PaperTableNotesStyle
            \item[] Values are the largest observed deviations; absolute deviations are reported unless marked as relative.
        \end{tablenotes}
    \end{threeparttable}
\end{table}
\FloatBarrier
\section{Compressed execution and radial tabulation}
\label{si:compressed-execution}

DPA4C compression combines fixed-size node tiling with a tabulated radial representation.
The execution schedule and memory scaling are defined together with the interpolation coefficients.

\subsection{Execution algorithm and memory scaling}
\label{si:compressed-algorithm}

The compressed path combines a tabulated approximation to the learned radial functions with a deployment-specific execution schedule and bounded-lifetime intermediate arrays.
Let \(B\) be the node-tile size, \(F_\tau\) the hidden width of MLP layer \(\tau\), \(F_{\max}=\max_\tau F_\tau\), and \(q=1\) for a one-layer MLP and \(q=2\) otherwise.
The deployed default is \(B=131{,}072\).

\paragraph{Execution order.}
The inference path consists of the following stages.

\begin{enumerate}[label=\arabic*.,leftmargin=2.2em]
    \item \emph{Offline radial tabulation and finite caches.}
          For every knot on \([0,\rcut]\), evaluate \([\bg(\rho),\bq(\rho)]\) and its first two derivatives, then store the six coefficients of the quintic Hermite interpolant on each interval.
          Evaluate and cache \(\bgamma\), \(\bbeta\) and \(\bU\) for all ordered type pairs, together with the initial-feature table, readout matrices, coupling tensors and calibration vectors.
          These objects depend on the trained model and the finite type set, not on \(N\) or \(N_{\mathrm{edge}}\).

    \item \emph{Canonical graph construction.}
          Compact the physical edges into destination-major order and store their source indices, Cartesian edge vectors and destination CSR row pointer.
          Construct a source CSR row pointer and a permutation from source order to the destination-major edge array.
          A contiguous range of destination atoms then owns one contiguous span of the edge stream.
          In the LAMMPS/Kokkos path, one warp processes the neighbor candidates of one center atom and uses an in-warp prefix count to place surviving neighbors in consecutive edge slots while retaining their candidate order.

    \item \emph{Forward evaluation of one node tile.}
          For a tile of at most \(B\) consecutive destination atoms, scan each destination CSR interval once.
          Evaluate the radial interpolant, ordered-pair modulation, envelope and Cartesian harmonics in registers, and accumulate the \(S+2\) feature state.
          Normalize the features and evaluate the polynomial invariants to form tile-local feature vectors of shape \((B,D_{\mathrm{out}})\).
          The only global-memory writes of this stage are the retained per-node state and the feature vectors; every per-edge quantity and every intermediate of the invariant evaluation stays in registers or shared memory.

    \item \emph{MLP forward pass and vector--Jacobian products.}
          Evaluate the MLP hidden layers and the scalar head for the tile, alternating hidden activations between \(q\) scratch slots and saving the layer pre-activations.
          Back-propagate the scalar seed through the MLP and overwrite the feature-vector workspace with \(\partial E/\partial\bD_i\).
          Apply the polynomial-invariant vector--Jacobian product; the resulting feature gradient overwrites the saved \(S+2\) state.

    \item \emph{Edge recomputation.}
          Revisit the same destination-sorted edge span and recompute the radial interpolation, type modulation, envelope and harmonics.
          Contract these quantities with the feature gradient to write the Cartesian edge derivative \(\partial E/\partial\br_{ij}\).
          The tile workspace is then reused for the next destination range.

    \item \emph{Energy, force and virial assembly.}
          Obtain the contiguous node interval of each frame from the prefix sum of the per-frame node counts.
          CUDA blocks first reduce disjoint slices of each interval into double-precision partial sums; a second fixed-order reduction combines the partials into the frame energy.
          This procedure requires neither a node-length frame-index array nor floating-point atomic additions to the frame accumulator.
          Traverse the destination and source CSR views to combine the two force contributions in Eq.~\eqref{eq:forces}; use the same edge derivatives and edge vectors to form the virial.
          Each node is written by one CSR reduction, so force and virial assembly requires no floating-point atomic accumulation.
\end{enumerate}

\paragraph{Memory decomposition.}
For one tile, the feature vector, the node-feature state, the saved MLP pre-activations and the MLP scratch require
\begin{equation}
    \mathcal{M}_{\mathrm{width}}
    =O\!\left(
    B\left[
        D_{\mathrm{out}}+(S+2)
        +\sum_{\tau=1}^{\Lambda}
        F_\tau +qF_{\max} \right] \right).
    \label{si:eq-tile-memory}
\end{equation}
The bracket contains every workspace term whose size changes with the feature or MLP width.
Because \(B\) is fixed independently of the simulated system, none of these terms scales with the full atom count.

The full-system state consists of source indices, source order, edge vectors and edge derivatives on the directed-edge axis; destination and source row pointers on the node axis; and atom types, atomic energies, forces and optional atomic virials on the node axis.
Its storage is therefore
\begin{equation}
    \mathcal{M}_{\mathrm{system}}
    =O(N_{\mathrm{edge}})+O(N)
    =O(N_{\mathrm{edge}}+N),
    \label{si:eq-system-memory}
\end{equation}
with constants set by graph and physical-output fields rather than by \(C_0\), \(L\), \(R\) or the MLP widths.
The model parameters and offline tables add a system-size-independent term.
Equations~\eqref{si:eq-tile-memory} and~\eqref{si:eq-system-memory} explain why increasing the size of a DPA4C variant changes arithmetic cost and tile-local storage without changing the dominant scaling of the largest simulation.

\subsection{Quintic Hermite interpolation of the radial table}
\label{si:quintic}

On the interval \([\rho_s,\rho_{s+1}]\) of width \(\Delta\), let \(y_s\), \(y'_s\) and \(y''_s\) denote the value and the first two derivatives of a tabulated channel at the left knot, and \(y_{s+1}\), \(y'_{s+1}\), \(y''_{s+1}\) the same quantities at the right knot.
The unique quintic reproducing all six is
\begin{equation}
    y(\rho)=\sum_{n=0}^{5}c_n\,(\rho-\rho_s)^{n},
    \qquad
    c_0=y_s,
    \qquad
    c_1=y'_s,
    \qquad
    c_2=\tfrac12 y''_s,
\end{equation}
with the remaining coefficients, writing \(\delta=y_{s+1}-y_s\),
\begin{equation}
    \begin{aligned}
        c_3 & =\frac{20\,\delta-\bigl(8y'_{s+1}+12y'_s\bigr)\Delta-\bigl(3y''_s-y''_{s+1}\bigr)\Delta^{2}}{2\Delta^{3}},    \\
        c_4 & =\frac{-30\,\delta+\bigl(14y'_{s+1}+16y'_s\bigr)\Delta+\bigl(3y''_s-2y''_{s+1}\bigr)\Delta^{2}}{2\Delta^{4}}, \\
        c_5 & =\frac{12\,\delta-6\bigl(y'_{s+1}+y'_s\bigr)\Delta+\bigl(y''_{s+1}-y''_s\bigr)\Delta^{2}}{2\Delta^{5}}.
    \end{aligned}
\end{equation}

\subsection{Radial-table spacing and numerical fidelity}
\label{si:compression-fidelity}

We compared seven radial-table spacings from 0.0005 to 0.05~\Ang{} for each of the five OMat24 models.
Each compressed model was paired with an uncompressed model exported from the same checkpoint, and all pairs were evaluated over the complete OMat24 validation split of 1,074,643 structures and 20,091,142 atoms.
Supplementary Table~\ref{si:tab-compression-fidelity} reports the prediction-discrepancy MAE and RMSE at four representative spacings, the finest, the 0.002~\Ang{} deployment default, 0.01~\Ang{} and the coarsest.
The spacings 0.001, 0.003 and 0.005~\Ang{} give discrepancies in the same range as the finest and are omitted.
The two models differ not only in the radial interpolation but also in graph layout, kernel fusion and floating-point summation order.
Their prediction differences therefore contain both interpolation error and rounding differences introduced by the distinct execution orders.
Between 0.0005 and 0.01~\Ang{}, the nearly constant, non-monotonic discrepancies indicate that the latter contribution is comparable to or larger than the remaining interpolation error, so a denser table does not necessarily lower the aggregate MAE or RMSE.
At 0.05~\Ang{}, the increase in force and stress RMSE for several variants marks the onset of a resolvable interpolation contribution, whereas the energy differences remain within the fine-spacing range.
Across all seven spacings, the largest absolute relative change in an MAE against the reference labels was \(5.23\times10^{-5}\%\).
Repeated complete-split evaluations at 0.002~\Ang{} reproduced the prediction-discrepancy MAEs to within 0.2\%.

\begin{table}[!htbp]
    \PaperTableStyle
    \caption{Compressed-minus-uncompressed prediction discrepancies at four radial-table spacings from 0.0005 to 0.05~\Ang{}.}
    \label{si:tab-compression-fidelity}
    \begin{threeparttable}
        \begin{tabular*}{\linewidth}{@{\extracolsep{\fill}}lcccccc@{}}
            \toprule
            \textbf{Model} &
            \multicolumn{2}{c}{\makecell{\textbf{Energy per atom}\\\textbf{(meV/atom)}}} &
            \multicolumn{2}{c}{\makecell{\textbf{Force}\\\textbf{(meV/\Ang)}}} &
            \multicolumn{2}{c}{\makecell{\textbf{Stress}\\\textbf{(meV/\Ang\textsuperscript{3})}}} \\
            \cmidrule(lr){2-3}
            \cmidrule(lr){4-5}
            \cmidrule(lr){6-7}
            & \textbf{MAE} & \textbf{RMSE} & \textbf{MAE} & \textbf{RMSE} & \textbf{MAE} & \textbf{RMSE} \\
            \midrule
\multicolumn{7}{@{}l}{\textit{Spacing 0.0005~\Ang{} (12,000 radial intervals)}} \\
            \addlinespace[2pt]
            Nano  & \(3.50\times10^{-4}\) & \(4.58\times10^{-4}\) & \(1.18\times10^{-3}\) & \(2.64\times10^{-3}\) & \(3.14\times10^{-5}\) & \(6.00\times10^{-5}\) \\
            Mini  & \(2.20\times10^{-4}\) & \(2.97\times10^{-4}\) & \(1.09\times10^{-3}\) & \(1.80\times10^{-3}\) & \(2.92\times10^{-5}\) & \(5.59\times10^{-5}\) \\
            Neo   & \(2.18\times10^{-4}\) & \(2.88\times10^{-4}\) & \(1.07\times10^{-3}\) & \(1.76\times10^{-3}\) & \(2.79\times10^{-5}\) & \(5.35\times10^{-5}\) \\
            Air   & \(1.99\times10^{-4}\) & \(2.62\times10^{-4}\) & \(9.30\times10^{-4}\) & \(1.57\times10^{-3}\) & \(2.54\times10^{-5}\) & \(4.53\times10^{-5}\) \\
            Plus  & \(1.98\times10^{-4}\) & \(2.61\times10^{-4}\) & \(1.01\times10^{-3}\) & \(1.67\times10^{-3}\) & \(2.75\times10^{-5}\) & \(5.00\times10^{-5}\) \\
            \midrule
\multicolumn{7}{@{}l}{\textit{Spacing 0.002~\Ang{} (3,000 radial intervals)}} \\
            \addlinespace[2pt]
            Nano  & \(3.52\times10^{-4}\) & \(4.60\times10^{-4}\) & \(1.18\times10^{-3}\) & \(3.05\times10^{-3}\) & \(3.13\times10^{-5}\) & \(6.04\times10^{-5}\) \\
            Mini  & \(2.21\times10^{-4}\) & \(2.97\times10^{-4}\) & \(1.09\times10^{-3}\) & \(1.80\times10^{-3}\) & \(2.92\times10^{-5}\) & \(5.59\times10^{-5}\) \\
            Neo   & \(2.18\times10^{-4}\) & \(2.89\times10^{-4}\) & \(1.07\times10^{-3}\) & \(1.76\times10^{-3}\) & \(2.80\times10^{-5}\) & \(5.35\times10^{-5}\) \\
            Air   & \(1.99\times10^{-4}\) & \(2.62\times10^{-4}\) & \(9.31\times10^{-4}\) & \(1.57\times10^{-3}\) & \(2.55\times10^{-5}\) & \(4.53\times10^{-5}\) \\
            Plus  & \(1.98\times10^{-4}\) & \(2.60\times10^{-4}\) & \(1.01\times10^{-3}\) & \(1.66\times10^{-3}\) & \(2.74\times10^{-5}\) & \(4.99\times10^{-5}\) \\
            \midrule
\multicolumn{7}{@{}l}{\textit{Spacing 0.01~\Ang{} (600 radial intervals)}} \\
            \addlinespace[2pt]
            Nano  & \(2.97\times10^{-4}\) & \(4.08\times10^{-4}\) & \(1.17\times10^{-3}\) & \(2.89\times10^{-3}\) & \(2.64\times10^{-5}\) & \(5.40\times10^{-5}\) \\
            Mini  & \(1.94\times10^{-4}\) & \(2.68\times10^{-4}\) & \(1.07\times10^{-3}\) & \(1.78\times10^{-3}\) & \(2.27\times10^{-5}\) & \(4.29\times10^{-5}\) \\
            Neo   & \(1.83\times10^{-4}\) & \(2.51\times10^{-4}\) & \(1.06\times10^{-3}\) & \(1.76\times10^{-3}\) & \(2.26\times10^{-5}\) & \(4.14\times10^{-5}\) \\
            Air   & \(1.77\times10^{-4}\) & \(2.37\times10^{-4}\) & \(9.14\times10^{-4}\) & \(1.54\times10^{-3}\) & \(1.94\times10^{-5}\) & \(3.53\times10^{-5}\) \\
            Plus  & \(1.65\times10^{-4}\) & \(2.24\times10^{-4}\) & \(9.82\times10^{-4}\) & \(1.63\times10^{-3}\) & \(2.24\times10^{-5}\) & \(4.00\times10^{-5}\) \\
            \midrule
\multicolumn{7}{@{}l}{\textit{Spacing 0.05~\Ang{} (120 radial intervals)}} \\
            \addlinespace[2pt]
            Nano  & \(2.65\times10^{-4}\) & \(4.33\times10^{-4}\) & \(1.19\times10^{-3}\) & \(6.71\times10^{-2}\) & \(2.29\times10^{-5}\) & \(1.30\times10^{-4}\) \\
            Mini  & \(1.88\times10^{-4}\) & \(2.60\times10^{-4}\) & \(1.22\times10^{-3}\) & \(2.01\times10^{-3}\) & \(2.88\times10^{-5}\) & \(5.60\times10^{-5}\) \\
            Neo   & \(2.19\times10^{-4}\) & \(3.50\times10^{-4}\) & \(1.46\times10^{-2}\) & \(3.20\times10^{-2}\) & \(5.33\times10^{-4}\) & \(1.33\times10^{-3}\) \\
            Air   & \(1.62\times10^{-4}\) & \(2.21\times10^{-4}\) & \(4.09\times10^{-3}\) & \(8.28\times10^{-3}\) & \(1.35\times10^{-4}\) & \(3.06\times10^{-4}\) \\
            Plus  & \(1.56\times10^{-4}\) & \(2.14\times10^{-4}\) & \(3.83\times10^{-3}\) & \(8.48\times10^{-3}\) & \(1.27\times10^{-4}\) & \(3.00\times10^{-4}\) \\
            \bottomrule
        \end{tabular*}
        \begin{tablenotes}[flushleft]
            \PaperTableNotesStyle
            \item[] Values are computed from compressed minus uncompressed predictions obtained from models exported from the same checkpoint.
            \item[] Energy errors are normalized by the number of atoms in each structure; force and stress errors are componentwise.
            \item[] MAE and RMSE are computed from the same primary evaluation. Both models reproduce the DPA4C MAEs against the reference labels in Table~\ref{tab:omat24-accuracy} at the reported precision.
        \end{tablenotes}
    \end{threeparttable}
\end{table}
\FloatBarrier

\section{Model, training and evaluation configurations}
\label{si:model-and-training}

\subsection{DPA4C model and training configurations}

The five DPA4C architectures differ in the three structural parameters \(C_0\), \(L\) and \(R\), and in the width of the atomic-energy MLP.
Table~\ref{si:tab-model-variants} summarizes these architecture choices and the optimizer settings shared across the OMat24, MatPES and OMol25 benchmarks.
All variants use the same 118-element type map.
OMol25 additionally supplies total charge and spin multiplicity through trainable embeddings, and its corresponding parameter counts are given in parentheses.
For OMat24, all five variants were trained on the published training set and evaluated on its validation set.
The MatPES benchmark uses the R2SCAN-2025.2 release~\cite{kaplan2025matpes}, with training and test splits containing 347,889 and 19,328 structures, respectively.
The OMol25 benchmark uses the 101.7-million-structure training split of the OMol-0 release and its out-of-distribution composition validation split~\cite{levine2025open}.
Tables~\ref{si:tab-omat24-training}, \ref{si:tab-matpes-training} and~\ref{si:tab-omol25-training} report the dataset-specific training configurations.
Parameters not listed there follow the shared settings in Table~\ref{si:tab-model-variants}.

\begin{table}[!htbp]
    \PaperTableStyle
    \caption{DPA4C model and shared training configurations.}
    \label{si:tab-model-variants}
    \begin{threeparttable}
        \begin{tabular*}{\linewidth}{@{\extracolsep{\fill}}l *{5}{c} @{}}
            \toprule
            \textbf{Hyperparameter} & \textbf{DPA4C-Nano} & \textbf{DPA4C-Mini} & \textbf{DPA4C-Neo} & \textbf{DPA4C-Air} & \textbf{DPA4C-Plus} \\
            \midrule
            Scalar channels \(C_0\)               & 8          & 32         & 64         & 64         & 128         \\
            Maximum degree \(L\)                  & 2          & 2          & 2          & 3          & 3           \\
            Shared radial modes \(R\)             & 0          & 0          & 0          & 4          & 4           \\
            Feature-state width \(S\)\tnote{a}    & 40         & 76         & 108        & 115        & 223         \\
            Feature-vector width \(D_{\mathrm{out}}\)\tnote{a} & 70         & 144        & 208        & 219        & 541         \\
            MLP hidden dim.                    & 96         & 192        & 256        & 256        & 384         \\
            MLP hidden layers                  & 3          & 3          & 3          & 3          & 3           \\
            \midrule
            Radial basis                           & Bessel     & Bessel     & Bessel     & Bessel     & Bessel      \\
            No. radial bases                       & 16         & 16         & 16         & 16         & 16          \\
            Activation func.                       & SiLU       & SiLU       & SiLU       & SiLU       & SiLU        \\
            Parameter precision                    & Float32    & Float32    & Float32    & Float32    & Float32     \\
            MLP residual time step             & False      & False      & False      & False      & False       \\
            \midrule
            Compile                                & True       & True       & True       & True       & True        \\
            bf16 AMP                               & False      & False      & False      & False      & False       \\
            \midrule
            Optimizer                              & HybridMuon & HybridMuon & HybridMuon & HybridMuon & HybridMuon  \\
            Muon mode                              & Slice      & Slice      & Slice      & Slice      & Slice       \\
            Magma Lite                             & True       & True       & True       & True       & True        \\
            Weight decay                           & $1\times10^{-3}$ & $1\times10^{-3}$ & $1\times10^{-3}$ & $1\times10^{-3}$ & $1\times10^{-3}$ \\
            \midrule
            Cutoff (\Ang)                          & 6          & 6          & 6          & 6          & 6           \\
            No. atom types                         & 118        & 118        & 118        & 118        & 118         \\
            Trainable parameters\tnote{b}          & \makecell{29,809\\(35,473)} & \makecell{145,513\\(200,169)} & \makecell{342,089\\(538,697)} & \makecell{433,673\\(630,281)} & \makecell{1,456,961\\(2,188,865)} \\
            \bottomrule
        \end{tabular*}
        \begin{tablenotes}[flushleft]
            \PaperTableNotesStyle
            \item[a] The feature-state width \(S\) and invariant-feature-vector width \(D_{\mathrm{out}}\) follow from \(C_0\) and \(L\), as specified in the main text and evaluated in Supplementary Table~\ref{si:tab-widths}.
            \item[b] Counts include the feature path, MLP and 118-element type-map parameters; parenthesized values include the trainable charge and spin embeddings used for OMol25.
        \end{tablenotes}
    \end{threeparttable}
\end{table}

\begin{table}[!htbp]
    \PaperTableStyle
    \caption{DPA4C training hyperparameters on OMat24.}
    \label{si:tab-omat24-training}
    \begin{threeparttable}
        \begin{tabular*}{\linewidth}{@{\extracolsep{\fill}}l *{5}{c} @{}}
            \toprule
            \textbf{Hyperparameter}       & \textbf{DPA4C-Nano}     & \textbf{DPA4C-Mini}     & \textbf{DPA4C-Neo}      & \textbf{DPA4C-Air}      & \textbf{DPA4C-Plus}     \\
            \midrule
            LR scheduler                  & Cosine                 & Cosine                 & Cosine                 & Cosine                 & Cosine                 \\
            Max.
            LR & $5\times10^{-3}$ & $4\times10^{-3}$ & $3\times10^{-3}$ & $3\times10^{-3}$ & $2\times10^{-3}$ \\ Min.
            LR & $1\times10^{-6}$ & $1\times10^{-6}$ & $1\times10^{-6}$ & $1\times10^{-6}$ & $1\times10^{-6}$ \\ Warmup ratio & 0.003 & 0.003 & 0.003 & 0.003 & 0.003 \\ Warmup start factor & 0.2 & 0.2 & 0.2 & 0.2 & 0.2 \\ Batch size (per GPU)\tnote{a} & $\lceil 50000/N\rceil$ & $\lceil 40000/N\rceil$ & $\lceil 30000/N\rceil$ & $\lceil 25000/N\rceil$ & $\lceil 15000/N\rceil$ \\ Training epochs & 12 & 12 & 12 & 12 & 12 \\ No.
            GPUs & 1 & 1 & 1 & 1 & 1 \\ \midrule Loss & MAE & MAE & MAE & MAE & MAE \\ Loss weights $(E,F,V)$ & 20, 20, 5 & 20, 20, 5 & 20, 20, 5 & 20, 20, 5 & 20, 20, 5 \\ Gradient max.
            norm            & 5                      & 5                      & 5                      & 5                      & 5                      \\
            \bottomrule
        \end{tabular*}
        \begin{tablenotes}[flushleft]
            \PaperTableNotesStyle
            \item[a] \(N\) denotes the number of atoms in each system; \(\lceil\cdot\rceil\) rounds up to the nearest integer.
        \end{tablenotes}
    \end{threeparttable}
\end{table}

\begin{table}[!htbp]
    \PaperTableStyle
    \caption{DPA4C training hyperparameters on MatPES.}
    \label{si:tab-matpes-training}
    \begin{threeparttable}
        \begin{tabular*}{\linewidth}{@{\extracolsep{\fill}}l *{5}{c} @{}}
            \toprule
            \textbf{Hyperparameter} & \textbf{DPA4C-Nano} & \textbf{DPA4C-Mini} & \textbf{DPA4C-Neo} & \textbf{DPA4C-Air} & \textbf{DPA4C-Plus} \\
            \midrule
            LR scheduler                                & WSD & WSD & WSD & WSD & WSD \\
            Max.
            LR & $5\times10^{-3}$ & $3\times10^{-3}$ & $2\times10^{-3}$ & $2\times10^{-3}$ & $1.5\times10^{-3}$ \\ Min.
            LR & $1\times10^{-6}$ & $1\times10^{-6}$ & $1\times10^{-6}$ & $1\times10^{-6}$ & $1\times10^{-6}$ \\ Warmup ratio & 0.003 & 0.003 & 0.003 & 0.003 & 0.003 \\ Warmup start factor & 0.2 & 0.2 & 0.2 & 0.2 & 0.2 \\ Decay ratio & 0.65 & 0.65 & 0.65 & 0.65 & 0.65 \\ Decay type & Cosine & Cosine & Cosine & Cosine & Cosine \\ Batch size (per GPU)\tnote{a} & $\lceil 10000/N\rceil$ & $\lceil 10000/N\rceil$ & $\lceil 10000/N\rceil$ & $\lceil 10000/N\rceil$ & $\lceil 10000/N\rceil$ \\ Training epochs & 500 & 500 & 500 & 500 & 500 \\ No.
            GPUs & 1 & 1 & 1 & 1 & 1 \\ \midrule Loss & MAE & MAE & MAE & MAE & MAE \\ Loss weights $(E,F,V)$ & 20, 20, 5 & 20, 20, 5 & 20, 20, 5 & 20, 20, 5 & 20, 20, 5 \\ Gradient max.
            norm                          & 5 & 5 & 5 & 5 & 5 \\
            \bottomrule
        \end{tabular*}
        \begin{tablenotes}[flushleft]
            \PaperTableNotesStyle
            \item[a] \(N\) denotes the number of atoms in each system; \(\lceil\cdot\rceil\) rounds up to the nearest integer.
        \end{tablenotes}
    \end{threeparttable}
\end{table}

\begin{table}[!htbp]
    \PaperTableStyle
    \caption{DPA4C training hyperparameters on OMol25.}
    \label{si:tab-omol25-training}
    \begin{threeparttable}
        \begin{tabular*}{\linewidth}{@{\extracolsep{\fill}}l *{5}{c} @{}}
            \toprule
            \textbf{Hyperparameter} & \textbf{DPA4C-Nano} & \textbf{DPA4C-Mini} & \textbf{DPA4C-Neo} & \textbf{DPA4C-Air} & \textbf{DPA4C-Plus} \\
            \midrule
            LR scheduler                  & WSD                    & WSD                    & WSD                    & WSD                      & WSD                    \\
            Max.
            LR & $5\times10^{-3}$ & $4\times10^{-3}$ & $3\times10^{-3}$ & $2.5\times10^{-3}$ & $1\times10^{-3}$ \\ Min.
            LR & $1\times10^{-6}$ & $1\times10^{-6}$ & $1\times10^{-6}$ & $1\times10^{-6}$ & $1\times10^{-6}$ \\ Warmup ratio & 0.003 & 0.003 & 0.003 & 0.003 & 0.003 \\ Warmup start factor & 0.2 & 0.2 & 0.2 & 0.2 & 0.2 \\ Decay ratio & 0.65 & 0.65 & 0.65 & 0.65 & 0.65 \\ Decay type & Cosine & Cosine & Cosine & Cosine & Cosine \\ Batch size (per GPU)\tnote{a} & $\lceil 100000/N\rceil$ & $\lceil 100000/N\rceil$ & $\lceil 100000/N\rceil$ & $\lceil 50000/N\rceil$ & $\lceil 18000/N\rceil$ \\ Training epochs & 25 & 25 & 25 & 25 & 25 \\ No.
            GPUs & 1 & 1 & 1 & 2 & 4 \\ \midrule Loss & MAE & MAE & MAE & MAE & MAE \\ Loss weights $(E,F,V)$ & 10, 5, 0 & 10, 5, 0 & 10, 5, 0 & 10, 5, 0 & 10, 5, 0 \\ Gradient max.
            norm            & 5                      & 5                      & 5                      & 5                        & 5                      \\
            \bottomrule
        \end{tabular*}
        \begin{tablenotes}[flushleft]
            \PaperTableNotesStyle
            \item[a] \(N\) denotes the number of atoms in each system; \(\lceil\cdot\rceil\) rounds up to the nearest integer.
        \end{tablenotes}
    \end{threeparttable}
\end{table}

\FloatBarrier

\subsection{Independent OMat24 evaluation of NEP89}
\label{si:nep89-omat24-evaluation}

The released NEP89 model \texttt{nep89\_20250409.txt}~\cite{liang2026nep89} was evaluated on the OMat24 validation set without subsampling.
Inference used the standalone \texttt{nep} executable from GPUMD \texttt{v5.5-59-g2f9d6e14}, built from source revision \texttt{2f9d6e14f3d78c87\allowbreak{}ddc92c41dfbbe5a8\allowbreak{}f1f0524f}.
The released model configuration contains 976,331 trainable parameters, reported as 0.976M in Table~\ref{tab:omat24-accuracy}.

NEP89 incorporates D3(BJ) interactions and was trained with D3 contributions added to the OMat24 subset~\cite{liang2026nep89}.
To reproduce this convention, PBE-D3(BJ) energy, force and virial contributions were evaluated for every validation structure using \texttt{dftd3} 1.3.1 through ASE 3.28.0~\cite{ase-paper,grimme2010consistent,grimme2011effect}.
The calculator configuration was \texttt{DFTD3(method="PBE", damping="d3bj")}.
The resulting contributions were added to the original OMat24 reference labels.
No separate dispersion correction was applied to the NEP89 predictions, and no post-hoc element-wise reference-energy shift was fitted or applied.

The per-atom energy MAE follows Eq.~\eqref{eq:energy-mae}.
The total-energy residual of each structure was divided by its atom count before taking the absolute value and averaging over structures.
Force errors were averaged over all atoms and Cartesian components.
Stress tensors were constructed as \(\bm\sigma=-\bXi/V\), and their errors were averaged over the nine components of the flattened \(3\times3\) tensors.
The reported units are meV/atom, meV/\Ang{} and meV/\Ang\(^{3}\), respectively.
Software versions and their roles in this evaluation are summarized in Supplementary Table~\ref{si:tab-nep89-software}.

\begin{table}[!htbp]
    \PaperTableStyle
    \caption{Software used for the independent OMat24 evaluation of NEP89.}
    \label{si:tab-nep89-software}
    \begin{tabular*}{\linewidth}{@{\extracolsep{\fill}} l c l @{}}
        \toprule
        \textbf{Software} & \textbf{Version or revision} & \textbf{Purpose} \\
        \midrule
        GPUMD             & \texttt{v5.5-59-g2f9d6e14}   & NEP89 inference \\
        ASE               & 3.28.0                        & Structures and calculator interface \\
        \texttt{dftd3}    & 1.3.1                         & PBE-D3(BJ) reference contributions \\
        \bottomrule
    \end{tabular*}
\end{table}

\FloatBarrier
\section{Single-GPU performance benchmarks}
\label{si:h20-benchmarks}

\subsection{Whole-step throughput, capacity and deployment ablations}
\label{si:benchmark-tables}

Supplementary Table~\ref{si:tab-inference-capacity} reports the system-size scan, Table~\ref{si:tab-node-tiling} isolates node tiling, and Table~\ref{si:tab-whole-step-ablation} gives the operator timings underlying the main-text throughput analysis.
Node tiling is isolated by evaluating the same five structural profiles with the tile disabled and with the default tile of 131,072 destination atoms, keeping the graph, model, precision and molecular-dynamics input fixed.
Throughput and the largest completed scan point are recorded for both paths.
The segmented energy reduction and the warp-per-center graph fill of Supplementary Note~\ref{si:compressed-algorithm} are timed at one million atoms.
The graph count pass remains thread-per-center because it writes no edge vectors and does not benefit from the same store coalescing.
The combined step-time changes add the energy-reduction and complete-graph-construction savings to the common post-optimization step time, giving 3.435~ms before rounding.
The run-to-run spread is approximately 1\%.

\begin{table}[!htbp]
    \PaperTableStyle
    \caption{Single-H20 whole-step throughput and system capacity.}
    \label{si:tab-inference-capacity}
    \begin{threeparttable}
        \begin{tabular*}{\linewidth}{@{\extracolsep{\fill}}L{2.5cm}C{2.45cm}C{2.45cm}C{2.65cm}C{2.65cm}@{}}
            \toprule
            \textbf{Model} & \makecell{\textbf{Throughput at}\\\textbf{1,000,000 atoms}\\\textbf{(M atoms/s)}\tnote{a}} & \makecell{\textbf{Saturated}\\\textbf{throughput}\\\textbf{(M atoms/s)}\tnote{b}} & \makecell{\textbf{Largest}\\\textbf{completed}\\\textbf{system (atoms)}\tnote{c}} & \makecell{\textbf{First failed}\\\textbf{system (atoms)}\tnote{c}} \\
            \midrule
            DPA4C-Nano & 16.7242 & 16.4830 & 14,051,520 & 14,522,880 \\
            DPA4C-Mini & 10.3745 & 10.1919 & 14,051,520 & 14,522,880 \\
            DPA4C-Neo  & 7.1156  & 7.0181  & 14,051,520 & 14,522,880 \\
            DPA4C-Air  & 3.8223  & 3.8059  & 14,051,520 & 14,522,880 \\
            DPA4C-Plus & 2.1640  & 2.1617  & 14,051,520 & 14,522,880 \\
            NEP89      & 8.6527  & 8.5674  & 10,455,280 & 11,036,032 \\
            \bottomrule
        \end{tabular*}
        \begin{tablenotes}[flushleft]
            \PaperTableNotesStyle
            \item[] DPA4C and NEP89 are measured in NVT simulations of the same diamond-carbon supercells using LAMMPS/Kokkos and GPUMD, respectively.
            \item[a] Throughput includes graph construction, model evaluation, force and virial assembly and integration; the values at one million atoms are individual scan points.
            \item[b] Saturated throughput summarizes the high-occupancy regime of each system-size scan.
            \item[c] The capacity columns report the largest completed system and the next scanned size that failed.
        \end{tablenotes}
    \end{threeparttable}
\end{table}

\begin{table}[!htbp]
    \PaperTableStyle
    \caption{Node-tiling ablation for the five DPA4C variants.}
    \label{si:tab-node-tiling}
    \begin{threeparttable}
        \begin{tabular*}{\linewidth}{@{\extracolsep{\fill}}L{2.3cm}C{2.6cm}C{2.6cm}C{2.8cm}C{2.8cm}@{}}
            \toprule
            \textbf{Model} & \makecell{\textbf{Throughput}\\\textbf{before}\\\textbf{(M atoms/s)}} & \makecell{\textbf{Throughput}\\\textbf{after}\\\textbf{(M atoms/s)}} & \makecell{\textbf{Largest}\\\textbf{completed}\\\textbf{before (atoms)}} & \makecell{\textbf{Largest}\\\textbf{completed}\\\textbf{after (atoms)}} \\
            \midrule
            DPA4C-Nano & 15.738 & 15.699 & 11,036,032 & 14,051,520 \\
            DPA4C-Mini & 9.988  & 9.813  & 9,524,736  & 14,051,520 \\
            DPA4C-Neo  & 6.919  & 6.783  & 8,000,000  & 14,051,520 \\
            DPA4C-Air  & 3.773  & 3.740  & 8,000,000  & 14,051,520 \\
            DPA4C-Plus & 2.139  & 2.146  & 5,510,880  & 14,051,520 \\
            \bottomrule
        \end{tabular*}
        \begin{tablenotes}[flushleft]
            \PaperTableNotesStyle
            \item[] Before evaluates the complete destination-node axis in one pass; after uses the default tile of 131,072 destination atoms.
            Each throughput pair uses the same model, graph, precision and molecular-dynamics input; capacity is the largest completed atom count.
        \end{tablenotes}
    \end{threeparttable}
\end{table}

\begin{table}[!htbp]
    \PaperTableStyle
    \caption{Whole-step operator ablations at one million atoms.}
    \label{si:tab-whole-step-ablation}
    \begin{threeparttable}
        \begin{tabular*}{\linewidth}{@{\extracolsep{\fill}}
                L{3.5cm}C{2.4cm}C{2.8cm}C{1.35cm}C{1.35cm}C{0.9cm}@{}} \toprule \textbf{Optimization} & \textbf{Model} & \textbf{Quantity} & \textbf{Before} & \textbf{After} & \textbf{Unit} \\ \midrule Segmented energy reduction & All variants & Component time & 1,758 & 17 & \(\mu\)s \\ Warp-per-center graph fill & All variants & Component time & 13,816 & 12,053 & \(\mu\)s \\ Complete graph construction & All variants & Component time & 24,106 & 22,412 & \(\mu\)s \\ \midrule Both operators & DPA4C-Nano & Whole-step time & 63.19 & 59.76 & ms \\ Both operators & DPA4C-Mini & Whole-step time & 100.36 & 96.93 & ms \\ Both operators & DPA4C-Neo & Whole-step time & 143.28 & 139.85 & ms \\ Both operators & DPA4C-Air & Whole-step time & 266.45 & 263.02 & ms \\ Both operators & DPA4C-Plus & Whole-step time & 466.24 & 462.81 & ms \\ \bottomrule \end{tabular*}
        \begin{tablenotes}[flushleft]
            \PaperTableNotesStyle
            \item[] The first three rows are direct component timings.
            Complete graph construction includes the unchanged count pass.
            The ``Both operators'' rows combine the 1.74-ms energy-reduction and 1.69-ms graph-construction savings with the shared post-optimization step time; they isolate the operator contributions and are not simultaneous paired measurements of the complete step.
        \end{tablenotes}
    \end{threeparttable}
\end{table}

\FloatBarrier
\subsection{Single-GPU crystal scans and empirical-potential references}
\label{si:classical-reference}

Supplementary Fig.~\ref{si:fig-v100-crystal-scan} compares complete NVT steps for diamond carbon and FCC copper on one Tesla V100-SXM2-16GB GPU.
Each periodic system is a near-cubic conventional-cell supercell.
Diamond carbon uses the eight-atom cell with lattice constant \(3.567~\Ang\), and FCC copper uses the four-atom cell with lattice constant \(3.615~\Ang\).
At the DPA4C and NEP89 model cutoff of \(6~\Ang\), the perfect crystals contain 158 and 78 neighbors per atom, respectively.
The empirical potentials retain their native cutoffs.
They are computational references rather than accuracy-matched baselines.

For every material--model pair, three independent Slurm allocations start at 128 requested atoms and double the requested count until the first recognized GPU out-of-memory failure.
Exactly three target-space bisections then refine the interval between the last successful and first failed requests.
Every request is converted to a near-cubic conventional-cell replication whose edge counts differ by at most one, and the exact realized atom count is retained in the result record.
Each point uses NVT dynamics at 300~K, a 1-fs time step, 10 warm-up steps and 100 timed steps.
The LAMMPS paths use a \(1.0~\Ang\) neighbor skin and check the neighbor list every step.
GPUMD uses its native neighbor handling.
The throughput curves report the arithmetic mean of the three complete-step measurements at each realized size.
The capacity values in Supplementary Tables~\ref{si:tab-v100-carbon-scan} and~\ref{si:tab-v100-copper-scan} are the median largest completed and first OOM systems.
All three repeats gave the same bracket for every model.
Timeouts and execution failures without a CUDA, Kokkos or engine OOM signature are not accepted as capacity boundaries.

\begin{figure}[!htbp]
    \centering
    \includegraphics[width=\linewidth]{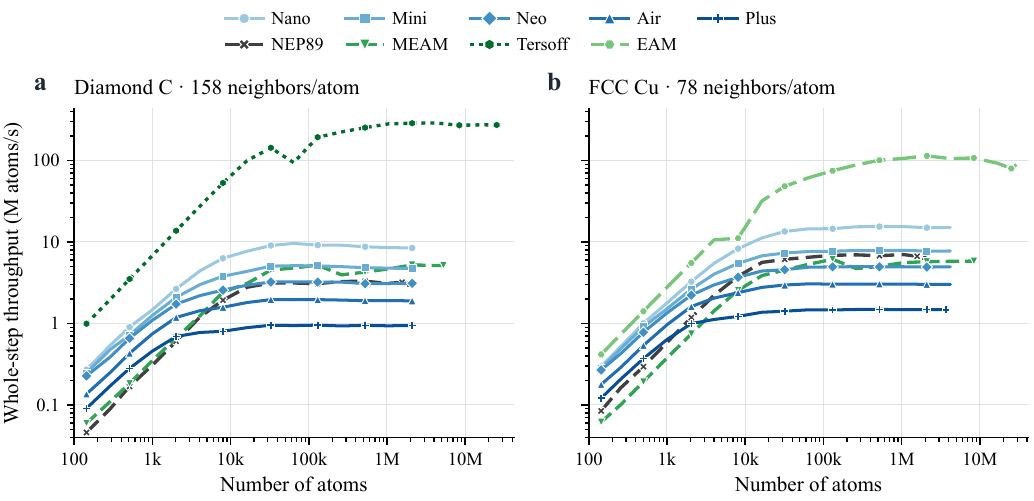}
    \caption{%
        Whole-step throughput across crystal sizes on a 16-GB NVIDIA Tesla V100-SXM2 GPU.
        \textbf{a}, Diamond carbon; \textbf{b}, FCC copper.
        Curves connect arithmetic means from three independent scans at each exact realized atom count.
        DPA4C and the empirical potentials use LAMMPS/Kokkos, whereas NEP89 uses native GPUMD; each empirical potential retains its native cutoff.
        The element-specific potentials provide computational references only, and predictive accuracy is outside this comparison.
    }
    \label{si:fig-v100-crystal-scan}
\end{figure}

\begin{table}[!htbp]
    \PaperTableStyle
    \caption{Diamond-carbon throughput and capacity on a 16-GB NVIDIA Tesla V100-SXM2 GPU.}
    \label{si:tab-v100-carbon-scan}
    \begin{threeparttable}
        \begin{tabular*}{\linewidth}{@{\extracolsep{\fill}}L{2.3cm}C{2.5cm}C{2.5cm}C{2.5cm}C{2.5cm}@{}}
            \toprule
            \textbf{Model} & \makecell{\textbf{Saturated}\\\textbf{throughput}\tnote{a}\\\textbf{(M atoms/s)}} & \makecell{\textbf{Atoms at}\\\textbf{saturated}\\\textbf{throughput}} & \makecell{\textbf{Largest}\\\textbf{completed}\\\textbf{system}} & \makecell{\textbf{First OOM}\\\textbf{system}} \\
            \midrule
            DPA4C-Nano & 9.6009   & 64,000    & 2,097,152  & 2,370,192  \\
            DPA4C-Mini & 5.1178   & 64,000    & 2,097,152  & 2,370,192  \\
            DPA4C-Neo  & 3.2380   & 64,000    & 2,097,152  & 2,370,192  \\
            DPA4C-Air  & 1.9627   & 32,768    & 2,097,152  & 2,370,192  \\
            DPA4C-Plus & 0.9545   & 130,000   & 2,097,152  & 2,370,192  \\
            NEP89      & 3.2980   & 524,800   & 1,968,624  & 2,097,152  \\
            \midrule
            MEAM       & 5.2302   & 2,097,152 & 5,268,024  & 5,767,200  \\
            Tersoff    & 288.7906 & 4,199,040 & 29,407,840 & 31,554,496 \\
            \bottomrule
        \end{tabular*}
        \begin{tablenotes}[flushleft]
            \PaperTableNotesStyle
            \item[] All capacities are operational bounds under the stated engine-native protocols, not hardware-independent limits.
            \item[a] The value is the largest three-allocation mean among successful scan points; throughput covers the complete NVT step.
        \end{tablenotes}
    \end{threeparttable}
\end{table}

\begin{table}[!htbp]
    \PaperTableStyle
    \caption{FCC-copper throughput and capacity on a 16-GB NVIDIA Tesla V100-SXM2 GPU.}
    \label{si:tab-v100-copper-scan}
    \begin{threeparttable}
        \begin{tabular*}{\linewidth}{@{\extracolsep{\fill}}L{2.3cm}C{2.5cm}C{2.5cm}C{2.5cm}C{2.5cm}@{}}
            \toprule
            \textbf{Model} & \makecell{\textbf{Saturated}\\\textbf{throughput}\tnote{a}\\\textbf{(M atoms/s)}} & \makecell{\textbf{Atoms at}\\\textbf{saturated}\\\textbf{throughput}} & \makecell{\textbf{Largest}\\\textbf{completed}\\\textbf{system}} & \makecell{\textbf{First OOM}\\\textbf{system}} \\
            \midrule
            DPA4C-Nano & 15.4357  & 1,048,576 & 4,203,216  & 4,719,120  \\
            DPA4C-Mini & 7.9040   & 1,048,576 & 4,203,216  & 4,719,120  \\
            DPA4C-Neo  & 4.9783   & 262,400   & 4,203,216  & 4,719,120  \\
            DPA4C-Air  & 3.0673   & 65,000    & 4,203,216  & 4,719,120  \\
            DPA4C-Plus & 1.4894   & 520,200   & 3,920,400  & 4,203,216  \\
            NEP89      & 7.0122   & 1,048,576 & 1,972,156  & 2,099,520  \\
            \midrule
            MEAM       & 6.1312   & 131,072   & 8,388,608  & 9,410,548  \\
            EAM        & 113.9039 & 2,099,520 & 29,356,080 & 31,522,396 \\
            \bottomrule
        \end{tabular*}
        \begin{tablenotes}[flushleft]
            \PaperTableNotesStyle
            \item[] All capacities are operational bounds under the stated engine-native protocols, not hardware-independent limits.
            \item[a] The value is the largest three-allocation mean among successful scan points; throughput covers the complete NVT step.
        \end{tablenotes}
    \end{threeparttable}
\end{table}

The calculations use LAMMPS 4 Jul 2026 with double-precision Kokkos, OpenMPI 5.0.10 and NVIDIA driver 550.163.01.
DPA4C uses the \texttt{deepmd/kk} pair style, a full neighbor list and Newton pair off.
MEAM, Tersoff and EAM use their Kokkos pair styles, half neighbor lists and Newton pair on.
NEP89 uses the native CUDA implementation in GPUMD revision \texttt{2f9d6e14}.
Carbon MEAM selects the C record in \texttt{library.meam} without a separate parameter file, and carbon Tersoff selects the C--C--C record in \texttt{SiC.tersoff}.
Copper EAM uses \texttt{Cu\_mishin1.eam.alloy}~\cite{mishin2001copper}.
Copper MEAM combines the Cu record in \texttt{library.meam} with \texttt{Cu.meam}.

Supplementary Listing~\ref{si:lst-v100-dpa4c} gives the DPA4C input shared by both crystals.
\texttt{lattice\_style}, \texttt{lattice\_constant}, \texttt{mass}, \texttt{element} and the replication counts are set from the selected crystal, and \texttt{model} is the packaged DPA4C variant.
The four empirical-potential inputs are reproduced in Supplementary Listings~\ref{si:lst-v100-carbon-meam}--\ref{si:lst-v100-copper-meam}.
LAMMPS is launched with one process using \texttt{-k on g 1 -sf kk}.
DPA4C adds \texttt{-pk kokkos neigh full newton off}, whereas the empirical potentials add \texttt{-pk kokkos neigh half newton on}.
Supplementary Listing~\ref{si:lst-v100-nep89} gives the GPUMD settings.
\texttt{model.xyz} contains the same generated crystal and exact realized atom count used by LAMMPS.

\noindent\begin{minipage}{\linewidth}
\begin{lstlisting}[
        style=LAMMPSInput,
        caption={DPA4C input for the single-V100 crystal scans.},
        label={si:lst-v100-dpa4c}
    ]
units           metal
boundary        p p p
atom_style      atomic
newton          off
atom_modify     map yes

lattice         ${lattice_style} ${lattice_constant}
region          box block 0 ${nx} 0 ${ny} 0 ${nz} units lattice
create_box      1 box
create_atoms    1 box
mass            1 ${mass}

pair_style      deepmd ${model}
pair_coeff      * * ${element}

neighbor        1.0 bin
neigh_modify    every 1 delay 0 check yes

velocity        all create 300.0 12345 mom yes rot yes dist gaussian
fix             integration all nvt temp 300.0 300.0 0.1
timestep        0.001

thermo          100
run             10
run             100
\end{lstlisting}
\end{minipage}

\noindent\begin{minipage}{\linewidth}
\begin{lstlisting}[
        style=LAMMPSInput,
        caption={Carbon MEAM input for the single-V100 crystal scan.},
        label={si:lst-v100-carbon-meam}
    ]
units           metal
boundary        p p p
atom_style      atomic
newton          on

lattice         diamond 3.567
region          box block 0 ${nx} 0 ${ny} 0 ${nz} units lattice
create_box      1 box
create_atoms    1 box
mass            1 12.011

pair_style      meam
pair_coeff      * * library.meam C NULL C

neighbor        1.0 bin
neigh_modify    every 1 delay 0 check yes

velocity        all create 300.0 12345 mom yes rot yes dist gaussian
fix             integration all nvt temp 300.0 300.0 0.1
timestep        0.001

thermo          100
run             10
run             100
\end{lstlisting}
\end{minipage}

\noindent\begin{minipage}{\linewidth}
\begin{lstlisting}[
        style=LAMMPSInput,
        caption={Carbon Tersoff input for the single-V100 crystal scan.},
        label={si:lst-v100-carbon-tersoff}
    ]
units           metal
boundary        p p p
atom_style      atomic
newton          on

lattice         diamond 3.567
region          box block 0 ${nx} 0 ${ny} 0 ${nz} units lattice
create_box      1 box
create_atoms    1 box
mass            1 12.011

pair_style      tersoff
pair_coeff      * * SiC.tersoff C

neighbor        1.0 bin
neigh_modify    every 1 delay 0 check yes

velocity        all create 300.0 12345 mom yes rot yes dist gaussian
fix             integration all nvt temp 300.0 300.0 0.1
timestep        0.001

thermo          100
run             10
run             100
\end{lstlisting}
\end{minipage}

\noindent\begin{minipage}{\linewidth}
\begin{lstlisting}[
        style=LAMMPSInput,
        caption={Copper EAM input for the single-V100 crystal scan.},
        label={si:lst-v100-copper-eam}
    ]
units           metal
boundary        p p p
atom_style      atomic
newton          on

lattice         fcc 3.615
region          box block 0 ${nx} 0 ${ny} 0 ${nz} units lattice
create_box      1 box
create_atoms    1 box
mass            1 63.546

pair_style      eam/alloy
pair_coeff      * * Cu_mishin1.eam.alloy Cu

neighbor        1.0 bin
neigh_modify    every 1 delay 0 check yes

velocity        all create 300.0 12345 mom yes rot yes dist gaussian
fix             integration all nvt temp 300.0 300.0 0.1
timestep        0.001

thermo          100
run             10
run             100
\end{lstlisting}
\end{minipage}

\noindent\begin{minipage}{\linewidth}
\begin{lstlisting}[
        style=LAMMPSInput,
        caption={Copper MEAM input for the single-V100 crystal scan.},
        label={si:lst-v100-copper-meam}
    ]
units           metal
boundary        p p p
atom_style      atomic
newton          on

lattice         fcc 3.615
region          box block 0 ${nx} 0 ${ny} 0 ${nz} units lattice
create_box      1 box
create_atoms    1 box
mass            1 63.546

pair_style      meam
pair_coeff      * * library.meam Cu Cu.meam Cu

neighbor        1.0 bin
neigh_modify    every 1 delay 0 check yes

velocity        all create 300.0 12345 mom yes rot yes dist gaussian
fix             integration all nvt temp 300.0 300.0 0.1
timestep        0.001

thermo          100
run             10
run             100
\end{lstlisting}
\end{minipage}

\noindent\begin{minipage}{\linewidth}
\begin{lstlisting}[
        style=LAMMPSInput,
        caption={NEP89 GPUMD settings for the single-V100 crystal scans.},
        label={si:lst-v100-nep89}
    ]
potential       nep89_20250409.txt
velocity        300
time_step       1
ensemble        nvt_nhc 300 300 100
dump_thermo     1000
run             10
run             100
\end{lstlisting}
\end{minipage}

\FloatBarrier
\section{Distributed V100 scaling}
\label{si:multigpu-scaling}

The distributed measurements use NVIDIA V100-SXM2-16GB GPUs, with 16 GPUs per node and one MPI process per GPU.
A complete node contains four NUMA-local groups of four GPUs: 0--3, 4--7, 8--11 and 12--15.
Every pair within a group is connected by two NVLinks.
Communication between groups uses the system interconnect.
The site mapping binds ranks to the corresponding CPU and GPU locality.
The first cross-node allocation is therefore 32 GPUs.

The V100 calculations use LAMMPS 4 Jul 2026, OpenMPI 5.0.10, CUDA 12.9 and NVIDIA driver 550.163.01.
The DeePMD runtime is version \texttt{3.2.0b1.
}\allowbreak\texttt{dev239+}\allowbreak\texttt{g048536d1a}, with PyTorch \texttt{2.13.0+}\allowbreak\texttt{cu126} and Python 3.13.15.
The launch environment sets \texttt{OMP\_NUM\_THREADS=2} for DPA4C, although this Kokkos path reports one active OpenMP thread per MPI process.
The standard DPA4-Mini path uses eight.
These settings remain fixed per GPU across their respective scaling series.

Each model follows its supported LAMMPS execution contract.
DPA4C uses \texttt{deepmd/kk}, a full neighbor list, Kokkos \texttt{newton off}, device communication and explicit device-aware reverse communication in the pair adapter.
DPA4-Mini uses the standard \texttt{deepmd} pair style without Kokkos, a full neighbor list and Newton pair forces.
The Newton setting follows the adapter and pair-style contract and is excluded from architectural interpretation.

DPA4C strong scaling uses the fixed 2,000,376-atom system defined in Methods.
Nano, Mini, Neo, Air and Plus use 400+4,000, 250+2,500, 150+1,500, 100+1,000 and 100+500 warm-up and timed steps, respectively.
The longer schedules for the smaller variants compensate for their lower per-step cost.
The per-GPU workload reaches about 3,907 atoms at 512 GPUs and 1,953 atoms at 1,024 GPUs.
Supplementary Fig.~\ref{si:fig-v100-strong-detail} gives the corresponding speedup and the LAMMPS timing fraction outside pair evaluation.
The latter includes neighbor construction, communication, fixes and other complete-step overheads.
It does not isolate MPI communication.

\begin{figure}[!htbp]
    \centering
    \includegraphics[width=\linewidth]{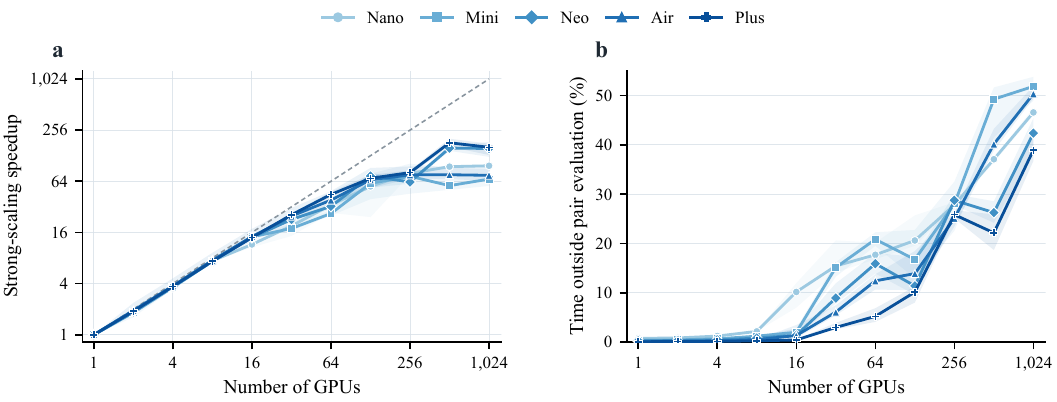}
    \caption{%
        Supplementary detail for DPA4C strong scaling of the fixed 2,000,376-atom system on 16-GB NVIDIA Tesla V100-SXM2 GPUs.
        \textbf{a}, Speedup relative to one GPU; the grey dashed line denotes linear speedup.
        \textbf{b}, Fraction of the complete LAMMPS step spent outside pair evaluation.
        Lines show arithmetic means and pale bands show the full range across ten independent allocations.
    }
    \label{si:fig-v100-strong-detail}
\end{figure}

The weak-scaling conventional-cell repetitions follow a balanced integer processor grid, giving every rank the same cubic \(63\times63\times63\)-cell local domain.
The global box is rectangular when the GPU count has no cubic factorization.
For example, 1,024 GPUs use an \(8\times8\times16\) processor grid and a \(504\times504\times1008\)-cell box.

Supplementary Table~\ref{si:tab-v100-weak-endpoints} gives the one- and 1,024-GPU endpoints.
DPA4-Mini uses its own cubic 5,832-atom-per-GPU local workload because its memory footprint does not admit the DPA4C workload.
Its 1,024-GPU system contains 5,971,968 atoms.
Its strong-scaling reference fixes 8,000 atoms on 1--32 GPUs, and both DPA4-Mini series use 100 warm-up and 500 timed steps.
The shared efficiency axis reports the fraction of each model's own single-GPU performance retained at fixed local work.
The unequal atom counts preclude an absolute-throughput comparison between the model families.

\begin{table}[!htbp]
    \PaperTableStyle
    \caption{Weak-scaling endpoints on 16-GB NVIDIA Tesla V100-SXM2 GPUs.}
    \label{si:tab-v100-weak-endpoints}
    \begin{threeparttable}
        {\small
            \begin{tabular*}{\linewidth}{@{\extracolsep{\fill}}lcccccc@{}}
                \toprule
                \textbf{Model} & \makecell{\textbf{Atoms}\\\textbf{per GPU}} & \makecell{\textbf{1-GPU}\\\textbf{rate}\\\textbf{(M atoms/s)}} & \makecell{\textbf{1,024-GPU}\\\textbf{rate}\\\textbf{(M atoms/s)}} & \makecell{\(\bm\eta_{\mathbf{w}}\)\\\textbf{(\%)}} & \makecell{\textbf{MD speed}\\\textbf{(ns/day)}} & \makecell{\textbf{Peak memory}\tnote{a}\\\textbf{(GiB)}} \\
                \midrule
                DPA4C-Nano & 2,000,376 & 8.995 & 7,673.4 & 83.3 & 0.324 & 13.5 / 13.6 \\
                DPA4C-Mini & 2,000,376 & 4.871 & 4,351.5 & 87.2 & 0.184 & 13.8 / 13.9 \\
                DPA4C-Neo  & 2,000,376 & 3.056 & 2,831.4 & 90.5 & 0.119 & 14.0 / 14.1 \\
                DPA4C-Air  & 2,000,376 & 1.908 & 1,755.6 & 89.9 & 0.074 & 14.1 / 14.2 \\
                DPA4C-Plus & 2,000,376 & 0.943 & 880.7   & 91.2 & 0.037 & 14.6 / 14.7 \\
                \midrule
                DPA4-Mini  & 5,832     & 0.0271 & 19.429 & 70.1 & 0.281 & 11.1 / 13.6 \\
                \bottomrule
            \end{tabular*}
        }
        \begin{tablenotes}[flushleft]
            \PaperTableNotesStyle
            \item[] Throughput and MD speed are medians of three independent allocations.
            \item[] \(\eta_{\mathrm{w}}\) is the 1,024-GPU weak-scaling efficiency of Eq.~\eqref{eq:weak-scaling}.
            \item[a] Peak memory lists the one- and 1,024-GPU per-process medians, separated by a slash.
        \end{tablenotes}
    \end{threeparttable}
\end{table}

A separate 5,000-step production trajectory at the 1,024-GPU DPA4C-Nano endpoint advances 5~ps in 1,342.73~s, sustaining 7.628 billion atoms/s and 0.322~ns/day in agreement with the median endpoint in Supplementary Table~\ref{si:tab-v100-weak-endpoints}.

Supplementary Fig.~\ref{si:fig-v100-comparison-detail}a reports DPA4-Mini strong scaling for a fixed cubic 8,000-atom system on 1, 2, 4, 8, 16 and 32 GPUs.
Its median efficiencies are 72.7, 64.9, 57.3, 45.8 and 34.6\%, respectively.
The final two points extend the curve to 500 and 250 atoms per GPU.
Supplementary Fig.~\ref{si:fig-v100-comparison-detail}b gives the full repeat ranges for the weak-scaling data shown in Fig.~\ref{fig:multigpu-scaling}c.

\begin{figure}[!htbp]
    \centering
    \includegraphics[width=\linewidth]{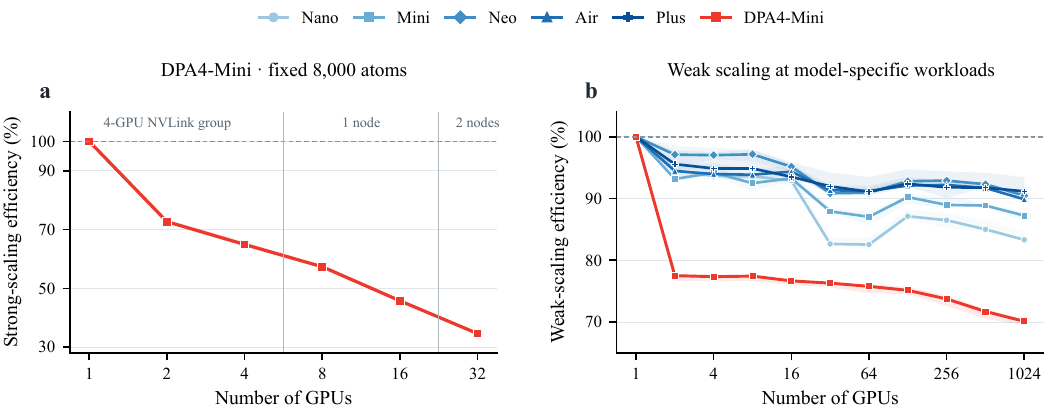}
    \caption{%
        Detailed DPA4C and DPA4-Mini scaling on 16-GB NVIDIA Tesla V100-SXM2 GPUs.
        \textbf{a}, DPA4-Mini strong-scaling efficiency at a fixed 8,000 atoms; region labels and vertical lines mark the four-GPU NVLink group, the remainder of the single node and the two-node allocation at 32 GPUs.
        \textbf{b}, Individual DPA4C and DPA4-Mini weak-scaling efficiencies shown in Fig.~\ref{fig:multigpu-scaling}c, with the full repeat ranges added here; DPA4C uses 2,000,376 and DPA4-Mini 5,832 atoms per GPU.
        Blue colors and markers denote DPA4C Nano through Plus as in Fig.~\ref{fig:multigpu-scaling}, and red squares denote DPA4-Mini.
        Lines show medians and pale bands show the full range of three allocations.
    }
    \label{si:fig-v100-comparison-detail}
\end{figure}